\documentclass[11pt]{article}

\usepackage{geometry,fullpage}
\usepackage{bm}
\usepackage{amsthm,amssymb,amsmath,amsfonts,rotate}
\usepackage{graphicx,algorithmic,algorithm}
\usepackage[]{datetime}
\usepackage{aliascnt}

\everymath{\color{MidnightBlack}}

\usepackage[pdftex,
backref=true,
plainpages = true,
pdfpagelabels,
hyperfootnotes=true,
pdfpagemode=FullScreen,
bookmarks=true,
bookmarksopen = true,
bookmarksnumbered = true,
breaklinks = true,
hyperfigures,
pagebackref,
urlcolor = magenta,
urlcolor = MidnightBlack,
anchorcolor = green,
hyperindex = true,
colorlinks = true,
linkcolor = black!30!blue,
citecolor = black!30!green
]{hyperref}

\RequirePackage[T1]{fontenc}
\RequirePackage[utf8]{inputenc}
\RequirePackage{amssymb}
\usepackage{alphabeta}
\DeclareUnicodeCharacter{2192}{\ifmmode\to\else\textrightarrow\fi}
\DeclareUnicodeCharacter{2203}{\ensuremath\exists}
\DeclareUnicodeCharacter{183}{\cdot}
\DeclareUnicodeCharacter{2200}{\forall}
\DeclareUnicodeCharacter{2264}{\leq}
\DeclareUnicodeCharacter{2265}{\geq}
\DeclareUnicodeCharacter{8614}{\mathbin{\mapsto}}
\DeclareUnicodeCharacter{8656}{\Leftarrow}
\DeclareUnicodeCharacter{8657}{\Uparrow}
\DeclareUnicodeCharacter{8658}{\Rightarrow}
\DeclareUnicodeCharacter{8659}{\Downarrow}
\DeclareUnicodeCharacter{8669}{\rightsquigarrow}
\newcommand{\eqdef}{\stackrel{{\scriptsize\rm def}}{=}}
\DeclareUnicodeCharacter{8797}{\eqdef}
\DeclareUnicodeCharacter{8870}{\vdash}
\DeclareUnicodeCharacter{8873}{\Vdash}
\DeclareUnicodeCharacter{22A7}{\models}
\DeclareUnicodeCharacter{9121}{\lceil}
\DeclareUnicodeCharacter{9123}{\lfloor}
\DeclareUnicodeCharacter{9124}{\rceil}
\DeclareUnicodeCharacter{9126}{\rfloor}
\DeclareUnicodeCharacter{9655}{\triangleright}
\DeclareUnicodeCharacter{9665}{\triangleleft}
\DeclareUnicodeCharacter{9671}{\diamond}
\DeclareUnicodeCharacter{9675}{\circ}
\DeclareUnicodeCharacter{10178}{\bot}
\DeclareUnicodeCharacter{10214}{\llbracket} 
\DeclareUnicodeCharacter{10215}{\rrbracket} 
\DeclareUnicodeCharacter{10229}{\longleftarrow}
\DeclareUnicodeCharacter{10230}{\longrightarrow}
\DeclareUnicodeCharacter{10231}{\longleftrightarrow}
\DeclareUnicodeCharacter{10232}{\Longleftarrow}
\DeclareUnicodeCharacter{10233}{\Longrightarrow}
\DeclareUnicodeCharacter{10234}{\Longleftrightarrow}
\DeclareUnicodeCharacter{10236}{\longmapsto}
\DeclareUnicodeCharacter{10238}{\Longmapsto} 
\DeclareUnicodeCharacter{10503}{\Mapsto}    
\DeclareUnicodeCharacter{10971}{\mathrel{\not\hspace{-0.2em}\cap}}
\DeclareUnicodeCharacter{65294}{\ldotp}
\DeclareUnicodeCharacter{65372}{\mid}

\usepackage{lmodern}
\usepackage{amsmath,ifthen,eurosym}
\usepackage{wrapfig,cite}
\usepackage{latexsym,amssymb,url}
\usepackage{subcaption}
\usepackage{cleveref}

\usepackage[svgnames]{xcolor}

\usepackage{color}

\definecolor{MidnightBlack}{rgb}{0.1,0.1,0.28}
\definecolor{MidnightBlue}{rgb}{0.1,0.1,0.44}
\definecolor{Black}{rgb}{0,0, 0}
\definecolor{Blue}{rgb}{0, 0 ,1}
\definecolor{Red}{rgb}{1, 0 ,0}
\definecolor{White}{rgb}{1, 1, 1}
\definecolor{Grey}{rgb}{.6, .6, .6}
\definecolor{Mygreen}{rgb}{.0, .7, .0}
\definecolor{Yellow}{rgb}{.55,.55,0}
\definecolor{Mustard}{rgb}{1.0, 0.86, 0.35}
\definecolor{applegreen}{rgb}{0.55, 0.71, 0.0}
\definecolor{darkturquoise}{rgb}{0.0, 0.81, 0.82}
\definecolor{celestialblue}{rgb}{0.29, 0.59, 0.82}
\definecolor{green_yellow}{rgb}{0.68, 1.0, 0.18}
\definecolor{crimsonglory}{rgb}{0.75, 0.0, 0.2}
\definecolor{darkmagenta}{rgb}{0.30, 0.0, 0.30}
\definecolor{internationalorange}{rgb}{1.0, 0.31, 0.0}
\definecolor{darkorange}{rgb}{1.0, 0.55, 0.0}

\usepackage{tikz}
\usetikzlibrary{calc,3d}
\usetikzlibrary{intersections,decorations.pathmorphing,shapes,decorations.pathreplacing,fit}
\usetikzlibrary{shapes.geometric,hobby,calc} 
\usetikzlibrary{decorations}
\usetikzlibrary{decorations.pathmorphing}
\usetikzlibrary{decorations.text}
\usetikzlibrary{shapes.misc}
\usetikzlibrary{decorations,shapes,snakes}

\tikzset{
	position/.style args={#1:#2 from #3}{
		at=($(#3)+(#1:#2)$)
	}
}

\tikzset{
  v:main/.style = {draw, circle, scale=0.65,fill=black,inner sep=0.7mm},
   v:ghost/.style = {inner sep=0pt,scale=1},
   e:main/.style = {thick},
 }

\usepackage{amsmath,latexsym,amsthm,amsfonts,amssymb,%
	 ifthen,color,wrapfig,hyperref,rotate,lmodern,aliascnt,%
	 datetime,graphicx,algorithmic,algorithm,enumerate,enumitem,%
	todonotes,cleveref,bm,subcaption,tabularx,colortbl,xspace%
}

\newcommand{\yes}{{\sf yes}}

\newcommand{\remove}[1]{}

\newcommand{\cupall}{\pmb{\pmb{\bigcup}}}

\newcounter{func}
\newcommand{\newfun}[1]{f_{\refstepcounter{func}\label{#1}\thefunc}}
\newcommand{\funref}[1]{\hyperref[#1]{f_{\ref*{#1}}}} 

\newcounter{con}

\newcommand{\conref}[1]{\hyperref[#1]{c_{\ref*{#1}}}} 

\newcommand{\tw}{{\sf tw}}

\usepackage{dsfont}

\usepackage{float}
\usepackage{tikz,pgf}
\usetikzlibrary{backgrounds}
\usetikzlibrary{calc,3d}

\tikzset{red node/.style={draw=red, circle, fill = red, minimum size = 4pt, inner sep = 0pt}}
\tikzset{yellow node/.style={draw=yellow, circle, fill = yellow, minimum size = 4pt, inner sep = 0pt}}
\tikzset{blue node/.style={draw=celestialblue, circle, fill =celestialblue, minimum size = 4pt, inner sep = 0pt}}
\tikzset{triangle/.style = { regular polygon, regular polygon sides=3, rotate=180}}
\tikzset{small red/.style={draw=red, triangle, fill = red, minimum size = 2pt, inner sep = 0pt}}

\tikzset{black node/.style={draw, circle, fill = black, minimum size = 3pt, inner sep = 0pt}}
\tikzset{small black node/.style={draw, circle, fill = black, minimum size = 3pt, inner sep = 0pt}}
\tikzset{model node/.style={draw=celestialblue, circle, fill = celestialblue, minimum size = 5pt, inner sep = 0pt}}
\tikzset{model node small/.style={draw=celestialblue, circle, fill = celestialblue, minimum size = 3pt, inner sep = 0pt}}
\tikzset{rep node/.style={draw=red, circle, fill = red, minimum size = 3pt, inner sep = 0pt}}
\tikzset{track node 1/.style={draw, circle, fill = black, minimum size = 2pt, inner sep = 0pt}}
\tikzset{track node 2/.style={draw=black!30!white, circle, fill = black!30!white, minimum size = 2pt, inner sep = 0pt}}
\tikzset{track node 3/.style={draw=black!10!white, circle, fill = black!10!white, minimum size = 2pt, inner sep = 0pt}}

\newcommand{\mynewtheorem}[2]{
	\newaliascnt{#1}{dummy}
	\newtheorem{#1}[#1]{#2}
	\aliascntresetthe{#1}
}

\theoremstyle{plain}
\mynewtheorem{theorem}{Theorem}
\mynewtheorem{proposition}{Proposition}
\mynewtheorem{corollary}{Corollary}
\mynewtheorem{lemma}{Lemma}

\theoremstyle{definition}
\mynewtheorem{definition}{Definition}
\theoremstyle{remark}
\mynewtheorem{sublemma}{Sublemma}
\mynewtheorem{claim}{Claim}
\mynewtheorem{remark}{Remark}
\mynewtheorem{fact}{Fact}
\mynewtheorem{conjecture}{Conjecture}
\mynewtheorem{question}{Question}
\mynewtheorem{observation}{Observation}
\mynewtheorem{problem}{Problem}
\mynewtheorem{note}{Note}

\newcommand{\FPT}{{\sf FPT}\xspace}

\newcommand{\obs}{{\bf obs}}

\newtheorem{innercustomgeneric}{\customgenericname}
\providecommand{\customgenericname}{}

\usepackage{titling}
\thanksmarkseries{arabic}
\newcommand*\samethanks[1][\value{footnote}]{\footnotemark[#1]}

\title{Block Elimination Distance\thanks{The first author was supported by the Spanish {\it Agencia Estatal de 
Investigacion} under project MTM2017-82166-P. 
The two last authors were supported  by   the ANR projects DEMOGRAPH (ANR-16-CE40-0028), ESIGMA (ANR-17-CE23-0010), ELIT (ANR-20-CE48-0008), and the French-German Collaboration ANR/DFG Project UTMA (ANR-20-CE92-0027).}}

\author{\bigskip Öznur Yaşar Diner\thanks{Computer Engineering Department, Kadir Has University, Istanbul, Turkey. {\tt oznur.yasar@khas.edu.tr}}$^{\ ,}$\thanks{Department of Mathematics, Universitat Politècnica de Catalunya, Barcelona, Spain.}\and Archontia C. Giannopoulou\thanks{Department of Informatics and Telecommunications, National and Kapodistrian University of Athens, Athens, Greece. \texttt{archontia.giannopoulou@gmail.com}} \and
Giannos Stamoulis\samethanks[4]$^{\ ,}$\thanks{LIRMM, Univ Montpellier, Montpellier, France.  \texttt{giannos.stamoulis@lirmm.fr}}
\and
Dimitrios  M. Thilikos\thanks{LIRMM, Univ Montpellier, CNRS, Montpellier, France. \texttt{sedthilk@thilikos.info}}}
\date{}

\begin{document}
\maketitle

\begin{abstract}
\noindent We introduce the parameter of {\sl block elimination distance} as a measure of how close
a graph is to some particular graph class.  Formally, given a graph class ${\cal G}$, the class   
${\cal B}({\cal G})$ contains all graphs whose blocks belong to ${\cal G}$
and the class ${\cal A}({\cal G})$ contains all graphs where the removal of a vertex creates a graph in ${\cal G}$. Given a hereditary graph class ${\cal G}$, we recursively define 
${\cal G}^{(k)}$ so that ${\cal G}^{(0)}={\cal B}({\cal G})$ and, if $k\geq 1$, ${\cal G}^{(k)}={\cal B}({\cal A}({\cal G}^{(k-1)}))$. The {\em block elimination distance} of a graph $G$ to a graph class ${\cal G}$ is the minimum $k$ such that $G\in{\cal G}^{(k)}$ and  can be seen as an analog  of the elimination distance parameter, 
defined in~[{\sl  J. Bulian and  A. Dawar. Algorithmica, 75(2):363–382, 2016}], with the difference that  connectivity is now replaced by biconnectivity. 
We show that, for every  non-trivial hereditary class ${\cal G}$,
the problem of deciding  whether $G\in{\cal G}^{(k)}$ is {\sf NP}-complete.
We focus on the case where ${\cal G}$ is minor-closed and we study the 
minor obstruction set  of  ${\cal G}^{(k)}$ i.e., the minor-minimal graphs not in ${\cal G}^{(k)}$.
We prove that the size of the obstructions of 
${\cal G}^{(k)}$ is upper bounded by some explicit function of $k$ and the maximum size of a minor obstruction of  ${\cal G}$.
This implies that the problem of deciding  whether $G\in{\cal G}^{(k)}$ is {\sl constructively} fixed parameter tractable, when parameterized by $k$.
Our results are based on a structural characterization of the obstructions of 
${\cal B}({\cal G})$, relatively to the obstructions of ${\cal G}$.
Finally, we give two graph operations  that generate members of ${\cal G}^{(k)}$ from members of ${\cal G}^{(k-1)}$ and  we prove that this set of operations is complete for the class ${\cal O}$ of outerplanar graphs. 
This yields the {\sl identification} of all members ${\cal O}\cap{\cal G}^{(k)}$, for every $k\in\mathbb{N}$ and every non-trivial minor-closed graph class ${\cal G}$.
\end{abstract}

\newpage 
\tableofcontents
\newpage 

\section{Introduction}

Graph distance parameters are typically introduced as measures 
of ``how close'' is a graph $G$ to some given graph class ${\cal G}$. 
One of the main motivating factors behind introducing such distance parameters is the following. Let $\mathcal{G}$ be a graph class on which a computational problem $\Pi$ is tractable and let $\mathcal{G}^{(k)}$ be the class of graphs with distance at most $k$ from $\mathcal{G}$, for some notion of distance. Our aim is to exploit the ``small'' distance of the graphs in $\mathcal{G}^{(k)}$ from $\mathcal{G}$ in order to extend the tractability of $\Pi$ in the graph class $\mathcal{G}^{(k)}$.
 This approach on dealing 
with computational problems is known as {\sl parameterization by distance from triviality}~\cite{GuoHN04}. 
Usually, a graph distance measure is defined by minimizing the number of 
modification operations that can transform a graph $G$ to a graph in ${\cal G}$. 

The most classic modification operation is the {\em apex extension} of a graph class ${\cal G}$, defined as 
${\cal A}({\cal G})=\{G\mid \exists v\in V(G)~~ G\setminus v\in {\cal G}\}$
and the associated  parameter, the  {\em vertex-deletion distance}  of $G$ to ${\cal G}$, is defined 
as $\min\{k\mid G\in {\cal A}^{k}({\cal G})\}$.
The vertex-deletion distance has been extensively studied. Other, popular variants of modification operations  involve edge removals/additions/contractions or combinations of them \cite{CrespelleDFG20asurv,FominSM15grap,BodlaenderHL14grap}.

\paragraph{Elimination distance.}
Bulian and Dawar in~\cite{BulianD16,BulianD17},  introduced the {\em elimination distance} of $G$ to a class ${\cal G}$ as follows:
$${\sf ed}_{\cal G}(G)=
\begin{cases}
 0 & G\in {\cal G}\\
 \max\{{\sf ed}_{\cal G}(C) \mid C\in {\sf cc}(G)\}& \mbox{if $G\not\in {\cal G}$ and $G$ is not connected  }\\
 1+\min\{{\sf ed}_{\cal G}(G\setminus v)\mid v\in V(G)\}& \mbox{if $G\not\in {\cal G}$ and $G$ is  connected  }
 \end{cases},$$
 where by $ {\sf cc}(G)$ we denote the connected components of $G$.
 Notice that the definition ${\sf ed}_{\cal G}$, apart from vertex deletions,  
 also involves the {\em connected closure} operation,
defined as ${\cal C}({\cal G})=\{G\mid  \forall C\in {\sf cc}(G),~C\in {\cal G}\}.$
Observe that ${\sf ed}_{\cal G}(G)=0$ iff 
$G\in {\cal G}\cup {\cal C}({\cal G})$, while, for $k>0$,  ${\sf ed}_{\cal G}(G)\leq k$ iff 
$G\in {\cal G}'\cup {\cal C}({\cal G}')$, where ${\cal G}'={\cal A}(\{G\mid {\sf ed}_{\cal G}(G)\leq k-1\})$. Therefore, ${\sf ed}_{\cal G}$ can be seen as 
a {\sl non-deterministic} counterpart  of the vertex-deletion operation where 
the operation ${\cal C}$ acts as {\sl the source of non-determinism}, that is, in each level of the recursion, the vertex deletion operation is applied to 
each of the connected components of the current graph. A motivation  
of Bulian and Dawar in~\cite{BulianD16} for introducing  ${\sf ed}_{\cal G}$ was the study of the {\sc Graph Isomorphism} Problem. Indeed, it is easy to see that there 
are constants $c_{α}$ and $c_{κ}$ such that 
if {\sc Graph Isomorphism} can be solved in $O(n^{c})$  time 
in some graph class ${\cal G}$, then it can be solved in time 
$O(n^{c+c_α})$ (resp. $O(n^{c+c_κ})$) in the graph class ${\cal A}({\cal G})$ (resp. ${\cal C}({\cal G})$) (see ~\cite{HopcroftT72isom,HopcroftT71,CorneilG70}).  This implies that 
 {\sc Graph Isomorphism} can be solved in $n^{O(k)}$
 steps in the class of graphs where ${\sf ed}_{\cal G}$ is bounded by $k$. In~\cite{BulianD16},  Bulian and Dawar 
 improved this implication for the class ${\cal G}_{d}$
of graphs of width at most $d$ and proved that {\sc Graph Isomorphism} can be solved in $f(k)\cdot n^{c_{d}}$ time in the 
class $\{G\mid {\sf ed}_{{\cal G}_{d}}(G)\leq k\}$ (here $c_{d}$ is a constant depending on $d$).  In other words, for every $d$,
{\sc Graph Isomorphism} is fixed parameter tractable (in short {\sf FPT}), 
when parameterized by ${\sf ed}_{{\cal G}_{d}}$.

\paragraph{Computing the elimination distance.} Typically,  the algorithmic results on ${\sf ed}_{\cal G}$ apply for instantiations of ${\cal G}$ that are hereditary, i.e., the removal of a vertex of a graph in ${\cal G}$ results to a graph that is again in ${\cal  G}$. Bulian and Dawar in~\cite{BulianD17} examined the 
case where ${\cal G}$ is minor-closed.
One may observe that containment in ${\cal G}$
is equivalent to the exclusion of the graphs in the minor-obstruction set of ${\cal G}$, that is the set ${\sf obs}(G)$ of the
 minor-minimal graphs not in ${\cal G}$.
 Also the minor-closed property is invariant under the operations ${\cal A}$ and ${\cal C}$, therefore the class $\{G\mid {\sf ed}_{\cal G}(G)\leq k\}$
is also minor-closed.
From the Robertson and Seymour theorem, ${\sf obs}(\{G\mid {\sf ed}_{\cal G}(G)\leq k\})$ is finite, and this implies, using the algorithmic results of~\cite{RobertsonS95b,KawarabayashiKR11thed}, that
for every minor-closed  class ${\cal G}$, deciding  whether ${\sf ed}_{\cal G}(G)\leq k$ is \FPT\ (parameterized by $k$) 
by an algorithm that runs in $f(k)\cdot n^2$ time. While this approach is not constructive in general, Bulian and Dawar in~\cite{BulianD17} proved that there is an algorithm that, with input ${\sf obs}({\cal G})$ and $k$, outputs the set ${\sf obs}(\{G\mid {\sf ed}_{\cal G}(G)\leq k\})$.
This makes the aforementioned 
$f(k)\cdot n^2$-time algorithm constructive in the sense that the function $f$ is computable.
An explicit estimation of this function $f$ 
can be derived from the recent results in~\cite{SauST20anfp,SauST21kapiI,SauST21kapiII}.
The computational complexity of ${\sf ed}_{\cal G}$ was also studied for different instantiations of ${\cal G}$. In~\cite{DLindermayrSV20}  Lindermayr, Siebertz, and Vigny considered the class ${\cal G}_{d}$ of graphs of degree at most $d$. They proved that, given  $k,d$, and a planar graph $G$, deciding  whether ${\sf ed}_{{\cal G}_{d}}(G)\leq k$
is \FPT\ (parameterized by $k$ and $d$) by designing an $f(k,d)\cdot n^{O(1)}$ time algorithm. Also, in~\cite{SridharanPASK21} the same result was proved without the planarity restriction.  Moreover, in~\cite{SridharanPASK21}, 
more general hereditary classes 
where considered: let ${\cal F}$ be some finite set of graphs and let ${\cal G}_{{\cal F}}$ be the class of graphs excluding all graphs in ${\cal F}$ as induced subgraphs. It was proved in \cite{SridharanPASK21}
that for every such ${\cal F}$ the problem that, given some graph $G$ and $k$, deciding  whether ${\sf ed}_{{\cal G}_{\cal F}}(G)\leq k$ is \FPT\ (parameterized by $k$) by designing an $f(k)\cdot n^{c_{d}}$ time algorithm, where $c_{d}$ is a constant depending on $d$ (see also~\cite{AgrawalR20onth} for earlier results).

\paragraph{Block elimination distance.}
We introduce a more general version of elimination distance where the source of non-determinism is {\sl biconnectivity} instead of 
connectivity. The recursive application of the vertex deletion operation is now done on the {\sl blocks} of the current graph instead of its  components.
That way, the
 {\em block elimination distance}
of a graph $G$ to  a graph class ${\cal G}$ is defined as
$${\sf bed}_{\cal G}(G)=
\begin{cases}
 0 & G\in {\cal G}\\
 \max\{{\sf bed}_{\cal G}(B) \mid B\in {\sf bc}(G)\}& \mbox{if $G\not\in {\cal G}$ and $G$ is not biconnected  }\\
 1+\min\{{\sf bed}_{\cal G}(G\setminus v)\mid v\in V(G)\}& \mbox{if $G\not\in {\cal G}$ and $G$ is  biconnected  }
 \end{cases},$$
 where by ${\sf bc}(G)$ we denote the blocks of the graph $G$.
We stress that the ``source of non-determinism'' in  the above definition is the   {\em biconnected closure} operation,
defined as ${\cal B}({\cal G})=\{G\mid  ∀ B\in {\sf bc}(G),~B\in {\cal G}\}.$ 

Notice that the above parameter is more general than ${\sf ed}_{\cal G}$ in the sense that 
it upper bounds ${\sf ed}_{\cal G}$ but it is not upper bounded by any function of ${\sf ed}_{\cal G}$: for instance, if $G$ is a connected graph whose blocks belong to ${\cal G}$, it follows that ${\sf bed}_{\cal G}(G)=0$, while  ${\sf ed}_{\cal G}(G)$ can be arbitrarily big.\footnote{It is easy to see that ${\sf ed}_{\cal G}(G)$  is logarithmically lower-bounded by the maximum number of cut-vertices in a path of $G$.}
Moreover, ${\sf bed}_{\cal G}$, 
can also serve as a measure for the distance to triviality in the same way as ${\sf ed}_{\cal G}$. For instance, there is a constant $c_{β}$ such that if {\sc Graph Isomorphism} can be solved in $O(n^{c})$  time 
in some graph class ${\cal G}$, then it can be solved in time $O(n^{c+c_β})$ in the graph class ${\cal B}({\cal G})$ (using standard techniques, see e.g., ~\cite{HopcroftT72isom,HopcroftT71,CorneilG70}).  This implies that 
 {\sc Graph Isomorphism} can be solved in $n^{O(k)}$
 steps in the class of graphs where ${\sf bed}_{\cal G}$ is bounded by $k$. 
 Clearly, all the problems studied so far on the elimination distance have their counterpart for the block elimination distance and this is a relevant line of research, as the new parameter is more general than its connected counterpart.

\paragraph{Our results.}
As a first step, we prove that if ${\cal G}$ 
is a non-trivial\footnote{A class is {\em non-trivial} if it contains at least one non-empty graph and is not the class of all graphs.} and  hereditary class, then 
deciding  whether ${\sf bed}_{\cal G}(G)\leq k$ is an {\sf NP}-complete problem (\autoref{np_trop}).
For our proof we certify \yes-instances by using an alternative definition of ${\sf bed}_{\cal G}$ that is 
based on an (multi)-embedding of $G$ in a rooted forest (\autoref{an_topote}).

We next focus our study on the case where ${\cal G}$ is minor-closed (and non-trivial).
As the operation ${\cal B}$ maintains minor-closedness, it follows 
 that the class ${\cal G}^{(k)}:=\{G\mid {\sf bed}_{\cal G}(G)\leq k\}$
 is minor-closed for every $k$, therefore 
%
%
%
%
for every minor-closed ${\cal G}$, deciding whether $G\in {\cal G}^{(k)}$
is {\sf FPT} (parameterized by $k$).
Following the research line of~\cite{BulianD17}, we make this result {\sl constructive} by proving that it is possible
to bound the size of the  obstructions of ${\cal G}^{(k)}$
by some explicit function of $k$ and the maximum size of the obstructions of ${\cal G}$. This bound is based on the results of~\cite{AdlerGK08comp,SauST21kapiI} (\autoref{label_enchantement}) and a structural characterization of ${\sf obs}({\cal B}(G)),$ in terms of   ${\sf obs}({\cal G})$,  implying that no obstruction of ${\cal B}(G)$ has size that is more than twice the maximum size of an obstruction of ${\cal G}$ (\autoref{label_sadistically}). 

 In~\autoref{op_rosete}  we take a closer look of the obstructions of  ${\cal G}^{(k)}$.  
We give two graph operations,  called {\sl parallel join} and {\sl triangular gluing},  that generate members of ${\cal G}^{(k)}$ from members of ${\cal G}^{(k-1)}$. 
This yields that the number of obstructions of  ${\cal G}^{(k)}$  is  at least doubly exponential on $k$.
Moreover, we prove that this set of operations is {\sl complete} for the class ${\cal O}$ of outerplanar graphs. This implies the  {\sl complete identification} of ${\cal O}\cap{\cal G}^{(k)}$, for every $k\in\mathbb{N}$ and every non-trivial minor-closed graph class ${\cal G}$. This yields that the number of obstructions of  ${\cal G}^{(k)}$  is  at least doubly exponential on $k$.

The paper concludes in \autoref{label_siberia} with some further observations and open problems.

\section{Definitions and preliminary results}
\label{an_topote}

\paragraph{Sets and integers.}\label{label_orgueilleuse}
We denote by $\mathbb{N}$ the set of non-negative integers.
Given two integers $p$ and $q,$ the set $[p,q]$ refers to the set of every integer $r$ such that $p\leq r\leq q.$
For an integer $p\geq 1,$ we set $[p]=[1,p]$ and $\mathbb{N}_{\geq p}=\mathbb{N}\setminus [0,p-1].$
For a set $S,$ we denote by $2^{S}$ the set of all subsets of $S$ and, given an integer $r\in[|S|],$
we denote by $\binom{S}{r}$ the set of all subsets of $S$ of size $r.$
If ${\cal S}$ is a collection of objects where the operation $\cup$ is defined,
then we denote $\cupall {\cal S}=\bigcup_{X\in {\cal S}}X.$
Given two sets $A,B$ and a function $f: A\to B,$
for every $X\subseteq A$ we use $f(X)$ to denote the set $\{f(x)\mid x\in X\}.$

\paragraph{Basic concepts on graphs.}\label{label_declaraciones}
All graphs considered in this paper are undirected, finite, and without loops or multiple edges.
We use $V(G)$ and $E(G)$ for the sets of vertices and edges of $G$, respectively.
For simplicity, an edge $\{x,y\}$ of $G$ is denoted by $xy$ or $yx$.
We say that $H$ is a {\em subgraph} of $G$ if $V(H)\subseteq V(G)$ and $E(H)\subseteq E(G)$.
For a set of vertices $S\subseteq V(G)$, we denote by $G[S]=(S,E(G)\cap{ S\choose 2})$ the subgraph of $G$ induced by the vertices from $S$. We also define $G\setminus S=G[V(G)\setminus S]$; we write $G\setminus v$ instead of $G\setminus\{v\}$ for a single vertex set. We say that the graph $H$ is an {\em induced subgraph} of a graph $G$ if $H=G[S]$ for some $S\subseteq V(G)$. 
Given $e\in E(G)$, we also denote $G\setminus e=(V(G),E(G)\setminus \{e\})$.
For a vertex $v$, we define the set of its {\em neighbors} in $G$
by $N_G(v)=\{u\mid vu\in E(G)\}$ and  ${\sf deg}_G(v)=|N_G(v)|$ denotes the \emph{degree} of $v$ in $G$. A vertex $v$ of $G$ is called {\em isolated} if ${\sf deg}_{G}(v)=0$. 
Given two graphs $G_{1}$ and $G_{2}$ we denote their disjoint union by $G_{1}+G_{2}$.
A graph $G$ is \emph{connected} if for every two vertices $u$ and $v$, $G$ contains a path whose end-vertices are $u$ and $v$ and it is {\em biconnected} if  for every two vertices $u$ and $v$, $G$ contains a cycle  containing the vertices $u$ and $v$. 
A {\em (bi)connected component} of $G$ is a  subgraph of $G$ that is maximally (bi)connected.
We denote by ${\sf cc}(G)$ the set of all connected components of $G$.
A {\em cut-vertex} of a graph $G$ is a vertex $x\in V(G)$ such that $|{\sf cc}(G)|<|{\sf cc}(G\setminus v)|$.
A {\em bridge} of a graph $G$ is a connected subgraph on two vertices $x,y$
and the edge $e=xy$ such that $|{\sf cc}(G)|<|{\sf cc}(G\setminus e)|$.
A {\em block} of a graph is either an isolated vertex, or a bridge of $G$, or a biconnected component of $G$. We also denote by ${\sf bc}(G)$ the set of all blocks of $G$ and we say that a graph $G$ is a {\em block-graph} if ${\sf bc}(G)=\{G\}$.

We use the term {\em graph class} (or simply {\em class}) for any set of graphs (this set might be finite or infinite). 
We say that a graph class is {\em non-trivial} if it contains at least one non-empty graph and does not contain all graphs.
We say that a class ${\cal G}$ is {\em hereditary}  if every induced subgraph  of a graph in $\mathcal{G}$ belongs also to ${\cal G}$.
Notice that both operations ${\cal A}$ and ${\cal B}$ maintain the property of being non-trivial and hereditary.
We denote by ${\cal E}$ the class of the edgeless graphs.

\paragraph{Some observations.}
In this paper we consider only classes that are non-trivial and hereditary. This implies that 
${\cal G}\subseteq {\cal A}({\cal G})$.  Notice that  this assumption is necessary as 
 $\{K_{1}\}\nsubseteq{\cal A}(\{K_{1}\})=\{K_{2}\}$ ($\{K_{1}\}$ is non-hereditary) and  $\{K_0\}\nsubseteq{\cal A}(\{K_0\})=\{K_{1}\}$ ($\{K_0\}$ is not non-trivial). Also the hereditarity of ${\cal G}$ implies that 
 ${\cal G}\subseteq {\cal B}({\cal G})$ and hereditarity is necessary for this as, for example,  $\{P_3\}\nsubseteq{\cal B}(\{P_3\})=\{K_0\}$.
However, ${\cal G}\subseteq {\cal B}({\cal G})$ also holds  for the two finite classes that are not non-trivial, i.e., ${\cal B}(\{\})=\{K_{0}\}$ and ${\cal B}(\{K_0\})=\{K_{0}\}$. We also exclude the class of all graphs as, in this case,  ${\cal A}$ and ${\cal B}$ do not generate new classes.

Given a $k\in\mathbb{N}$, we define ${\cal G}^{(k)}=\{G\mid {\sf bed}_{\cal G}(G)\leq  k\}.$ Observe that, according to the definition of  ${\sf bed}_{\cal G}$, ${\cal G}^{(0)}={\cal G}\cup {\cal B}({\cal G})$ while, for $k>0$,  ${\cal G}^{(k)}={\cal A}({\cal G}^{(k-1)})\cup {\cal B}({\cal A}({\cal G}^{(k-1)}))$. This, together with  the fact that  ${\cal G}\subseteq {\cal B}({\cal G})$ implies  that 
\vspace{-8mm}

\begin{eqnarray}
{\cal G}^{(k)}= \overbrace{{\cal B}({\cal A}(\cdots {\cal B}({\cal A}}^{\text{$k$ times}}({\cal B}({\cal G})))\cdots)).\label{label_effectiveness}
\end{eqnarray}
Observe also that for every non-trivial and hereditary  class ${\cal G}$, ${\cal B}({\cal G})={\cal B}({\cal B}({\cal G}))$.  This implies that ${\sf bed}_{\cal G}$
and ${\sf bed}_{{\cal B}({\cal G})}$ are the same parameter.

\paragraph{An alternative definition.}
\label{dr_oteros}

A {\em rooted forest} is a pair $(F,R)$ where $F$ is an acyclic graph and 
$R\subseteq V(F)$ such that each connected component of $F$ contains exactly one vertex of $R$, its {\em root}. A vertex $t\in V(F)$ is a {\em leaf} of $F$ if either $t\in R$ and ${\sf deg}_{F}(t)=0$  or $t\not\in R$ and ${\sf deg}_{F}(t)=1$. We use $L(F,R)$ in order to denote the leaves of $(F,R)$.
Given $t,t'\in V(F)$ we say that $t\leq_{F,R}t'$ if there is a path from $t'$ to some root in $R$ that contains $t$. If neither $t\leq_{F,R}t'$ nor $t'\leq_{F,R}t$ then we say that $t$ and $t'$ are {\em incomparable} in $(F,R)$.
A {\em $(F,R)$-antichain} is a non-empty set $C$ 
of pairwise incomparable vertices of $F$. An  $(F,R)$-antichain is {\em non-trivial} if it contains at least two elements.

Given a vertex $t\in V(F)$, we define its {\em descendants} in $(F,R)$
as the set ${\sf d}_{F,R}(t)=\{t'\in V(F)\mid t\leq_{F,R}t'\}.$
The {\em children} of a vertex $q\in V(F)$ in $(F,R)$ are the descendants of $q$ in $(F,R)$ that are adjacent to $q$ in $F.$
The {\em depth} of  a rooted forest  $(F,R)$ is the maximum number of vertices 
in a path 
between a leaf and the root of the connected component of $F$ where this leaf belongs.

Let ${\cal G}$ be a non-trivial hereditary class and let $G$ be a graph.
Let $(F,R,τ)$ be a triple consisting of  a rooted forest $F$ whose root set is $R$ and a function 
$τ: V(G)\to 2^{V(F)}$.
Given a vertex set $S\subseteq V(F)$, we set $τ^{-1}(S)=\{v\in V(G)\mid \tau(v)\cap S\neq\emptyset\}$.  
Also, for every $t\in V(F),$ we define $G_{t}=G[τ^{-1}({\sf d}_{F,R}(t))]$.

We say that a triple $(F,R,τ)$ is a {\em ${\cal G}$-block tree layout} of $G$
if the following hold:
\begin{enumerate}
\item[(1)]  for every $v\in V(G),$ $τ(v)$ is an $(F,R)$-antichain,
\item[(2)] for every $t\in V(T)$, $G_{t}$ is a block-graph,
\item[(3)] if $t\not\in L(F,R)$, then   $|τ^{-1}(\{t\})|=1$ and $G_{t}\not\in{\cal G}$ or ,
\item[(4)] if $t\in L(F,R)$, then $G_{t}\in {\cal G}$ and  
\item[(5)] for every non-trivial $(F,R)$-antichain $C$, the graph $\cupall\{G_{t}\mid t\in C\}$ is not biconnected.
\end{enumerate}
The {\em depth} of the ${\cal G}$-block tree layout $(F,R,τ)$ is equal to the depth of 
the rooted forest $(F,R)$.

%
\begin{lemma}[$\star$]
\label{label_exaggeration}
Let ${\cal G}$ be a non-trivial hereditary class and let $G$ be a graph.
Then the minimum depth of a ${\cal G}$-block tree layout of $G$ is equal to ${\sf bed}_{\cal G}(G)-1$.
\end{lemma}

\begin{proof}
Assume that  $(F,R,τ)$ is a ${\cal G}$-block tree
of depth $k+1\geq 1$. We use induction on $k$.
Notice first that if $k=0$, then 
$V(F)=L(F,R)=R$. From (2) and (5), ${\sf bc }(G)=\{G_{r}\mid r\in R\}$. From (4), for every $r\in R$, $G_{r}\in{\cal G},$ therefore  $G\in {\cal B}({\cal G})={\cal G}^{(0)}$. 

Suppose now that $k\geq 1$ and consider some $r\in R$. If $r\in L(F,R)$, then, because of  (4), 
$G_{r}\in {\cal G}\subseteq {\cal G}^{(k-1)}\subseteq {\cal A}({\cal G}^{(k-1)})$.
Suppose now that $r\not\in L(F,R)$. Then, from (3),
$|τ^{-1}(\{r\})|=1$ and we define $v_{r}$ so that  $τ^{-1}(\{r\})=\{v_{r}\}$. We also set $G_{r}^{-}=G_{r}\setminus v_r$.
Let $F_{r}=F[{\sf d}_{F,R}(r)]\setminus r$, $R_{r}=N_{F}(r)$,
and $\tau_{r}=\{(v,\tau(v)\cap V(F_{r}))\mid v\in V(G_{r}^{-})\}$
and observe that $(F_{r},R_{r},\tau_{r})$ is a ${\cal G}$-block tree layout  of $G_{r}^{-}$ of depth $k-1$. By the induction hypothesis $G_{r}^{-}\in {\cal G}^{(k-1)}$, therefore $G_{r}\in {\cal A}(G^{(k-1)})$.  Recall now that $R$ is an $(F,R)$-antichain, therefore from (2) and (5), we have that 
${\sf bc}(G)=\{G_{r}\mid r\in R\}$. This together with the fact that 
for all $r\in R$, $G_{r}^{-}\in {\cal G}^{(k-1)}$ imply that $G\in{\cal B}({\cal A}({\cal G}^{(k-1)}))$, therefore, $G\in{\cal G}^{(k)}$.

Suppose now that $G\in {\cal G}^{(k)}$ for some $k\geq 0$.
Again we use induction on $k$. In case $k=0$, observe that 
$G\in {\cal B}(G)$, therefore every block of $G$ belongs to ${\cal G}$.
We consider a rooted forest $(F,R)$ consisting of isolated 
vertices, one, say $R_{B}$, for each block $B$ of $G$.
We also set up a function $τ: V(G)\to 2^{V(F)}$ such that, for each vertex $v$ of $G$, $\tau(v)=\{R_{B}\mid v\in V(B)\}$.
Observe that $(F,R,\tau)$ is a ${\cal G}$-block tree layout of ${\cal G}$ of depth 1.

Suppose now that $k\geq 1$ and let $B\in{\sf bc}(G)$.
As $G\in {\cal G}^{(k)}$, it follows that $B\in {\cal A}({\cal G}^{(k-1)})$, therefore $B$ contains a vertex $a_{B}$
such that $B^{-}=B\setminus a_{B}\in {\cal G}^{(k-1)}$.
From the induction hypothesis, $B^{-}$ has a ${\cal G}$-block tree layout $(F_{B^{-}},R_{B^{-}},τ_{B^{-}})$ of depth $k$.
We use $(F_{B^{-}},R_{B^{-}},τ_{B^{-}})$ for all $B\in{\sf bc}(G)$ in order to construct a  ${\cal G}$-block tree layout $(F,R,\tau)$ of $G$  as follows. $F$ is constructed by first taking the disjoint union of all forests in $\{F_{B^{-}}\mid B\in{\sf bc}(G)\}$ then adding one new root vertex  $r_{B}$ for each $B\in{\sf bc}(G)$
and, finally,  making $r_{B}$ adjacent with all the vertices of $R_{B^{-}}$. We also set $R=\{r_{B}\mid B\in{\sf bc}(G)\}$.
For the construction of $\tau$, if $v\in\{B^{-}\mid B\in{\sf bc}(G)\},$ 
then $τ(v)=\cupall\{τ_{B^{-}}(v)\mid B\in{\sf bc}(G)\}$ and if $v\in\{a_{B}\mid B\in{\sf bc}(G)\}$, then $\tau(v)=\{R_{B^{-}}\mid v\in B\}$.
The result follows, as  $(F,R,\tau)$ has depth $k+1$.
 \end{proof}

\section{{\sf NP}-completeness}
\label{np_trop}
We consider the following family of problems, each defined by some non-trivial and hereditary graph class ${\cal G}$. We say that a class ${\cal G}$ is {\em polynomially decidable}  if there exists an algorithm that, given an $n$-vertex graph $G$, decides whether $G\in {\cal G}$ in polynomial, on $n$, time.

\begin{center}
	\fbox{
	\begin{minipage}{13.5cm}
			\noindent{\sc Block Elimination Distance to ${\cal G}$ ({\sc ${\cal G}$-BED})}

\noindent{\bf Instance:}  A graph $G$ and a non-negative integer $k$.

\noindent{\bf Question:}   Is the block elimination distance of $G$ to ${\cal G}$  at most $k$?		\end{minipage}
	}
\end{center}

\begin{lemma}[$\star$]
\label{label_intermediary}
For every polynomially decidable, non-trivial, and hereditary graph class ${\cal G}$, the problem {\sc ${\cal G}$-BED} is {\sf NP}-complete.
\end{lemma}

\begin{proof}
Given a graph $G$ and a non-negative integer $k$ and using \autoref{label_exaggeration},  we certify that ${\bf bed}_{\cal G}(G)\leq k$ by 
a  ${\cal G}$-block tree layout $(F,R,\tau)$ of $G$  of depth at most $k+1$.
This, together with the fact that  ${\cal G}$ is polynomially decidable, implies that 
{\sc ${\cal G}$-BED} belongs to ${\sf NP}$.

We next prove that  {\sc ${\cal G}$-BED} is {\sf NP}-hard.
Our first step is to prove that  {\sc ${\cal E}$-BED} is {\sf NP}-hard. Notice that, in this case, 
 conditions (2) and (4) of the definition of a ${\cal E}$-block tree layout 
imply that if $t\in L(F,R)$ then $G_{t}=K_{1}$. This, in turn, implies that 
$\{\tau(v)\mid v\in V(G)\}$ is a partition of $V(F)$.
\medskip

%
%
%

We present a reduction to {\sc ${\cal E}$-BED}  from the following {\sf NP}-hard problem:

\begin{center}
	\fbox{
	\begin{minipage}{13.5cm}
			\noindent{\sc Balanced Complete Bipartite Subgraph} (BCBS)

\noindent{\bf Instance:}  A bipartite graph $G$ with partition $V_{1},V_{2}$
and a positive integer $k$.

\noindent{\bf Question:}  Are there $W_{i}\subseteq V_{i},i\in[2]$, such that $G[W_{1}\cup W_{2}]$ is a complete bipartite graph  and $|W_{1}|=|W_{2}|=k$?
		\end{minipage}
	}
\end{center}

Let $G$ be a bipartite graph with partition $V_1, V_2$.
Let $n=|V(G)|$, $ξ=2n-4k$, and $k'=2n+ξ-2k$. Also keep in mind that $2\xi + 2k +1 = k'+1$.
For each vertex $v\in V(G)$, we consider a new vertex $v'$ and we denote by $V'$ this set of $n$ new vertices (i.e., $V' = \{v'\mid v\in V(G)\}$).
We consider the graph
$$G^\bullet = (V(G)\cup V', E(G)\cup \bigcup_{\{u,v\}\in E(G)}(\{u,v'\} \cup \{u',v\} \cup \{u',v'\})).$$
Then, we consider the graph $G^\star$
obtained by $\overline{G^\bullet}$ 
after adding a set $\hat{V}$ of $ξ+1$ new vertices
and make them adjacent 
with all the vertices in $V(G^\bullet)$.
We set $n^\star=|V(G^\star)|$ and we observe that 
$n^\star=2n+ξ+1$.
Also, for each $i\in[2]$, we set $V_i ' = \{v'\in V'\mid v\in V_i\}$ and $V_i^\star = V_i \cup V_i '$.
In what follows, we prove that $(G,k)$ is a \yes-instance of BCBS iff $(G^\star, k')$ is a \yes-instance of  {\sc ${\cal E}$-BED}.
We begin by proving that if $(G,k)$ is a \yes-instance of BCBS, then $(G^\star, k')$ is a \yes-instance of  {\sc ${\cal E}$-BED}.

%
Suppose that $(G,k)$ is a \yes-instance of BCBS.
Therefore, there exist $W_{i}\subseteq V_i, i\in [2]$ such that $G[W_1\cup W_2]$ is a complete bipartite graph and $|W_i| = k, i\in [2]$.
For each $i\in[2]$, we set $W_i'=\{v'\in V'\mid v\in W_i\}$ and $W_i^\star = W_i \cup W_i '$ and we observe that the graph $G^\bullet [W_1^\star \cup W_2^\star]$ is a complete bipartite graph whose parts are $W_1^\star$ and $W_2^\star$, each of size $2k$.
We now aim to define a triple $(F,R,\tau)$ that certifies that ${\bf bed}_{\cal E}(G^\star)\leq k'$.
To define $F$, let $P$ be an $(r,q)$-path of $2\xi +1$ vertices and $P_1$ (resp. $P_2$) be an $(a_1,\ell_1)$-path (resp. $(a_2, \ell_2)$-path) of $2k$ vertices such that $P$, $P_1$, and $P_2$ are pairwise vertex-disjoint.
We set $F$ to be the graph obtained from $P\cup P_1\cup P_2$ by adding the edges $q a_1$ and $q a_2$ and observe that $F$ is a tree of depth $2\xi + 2k +1 = k'+1$.
%
We now consider the triple $(F,R,τ)$, where $R=\{r\}$ and $τ$ is a function mapping each vertex of $\hat{V} \cup (V_{1}^\star \setminus W_1^\star)\cup (V_{2}^\star \setminus W_2^\star)$
to a unique vertex of $P$ and each vertex of $W_i^\star$ to a unique vertex of $P_{i}$, for $i\in[2]$.
It is easy to verify that $τ$  satisfies the properties (1) to (5) of the definition of ${\cal G}$-block tree layout, where ${\cal G} = {\cal E}$, and therefore, since $(F,R,τ)$ has depth $k'+1$, by \autoref{label_exaggeration},
we have that $(F,R,\tau)$ certifies that ${\bf bed}_{\cal E}(G^\star)\leq k'$.

What remains now is to prove that if ${\bf bed}_{\cal E}(G^\star) \leq k'$, then $(G,k)$ is a \yes-instance of BCBS.
%
Towards this, we argue that the following holds.\medskip

\noindent{\em Claim:} There are sets $B_i\subseteq V_i^\star$ such that $|B_i|\geq 2k-1, i\in[2]$, and there is no edge between vertices of $B_1$ and $B_2$ in $G^\star$.\medskip

\noindent{\em Proof of Claim:}
Assume that the ${\cal E}$-block tree layout $(T,R,τ)$ certifies that  ${\bf bed}_{\cal E}(G^\star)\leq k'$.
Observe that, since $G^\star$ is connected, $T$ is connected and $R$ is a singleton.
Let $r\in V(T)$ such that $R= \{r\}$.
Keep in mind that $T$ has depth at most $k'+1$.
Let $P$ be the $(r,q)$-path in $T$ where only $q$ has degree more than two in $T$.
Let $C=τ^{-1}(V(P))\cap (V_{1}^\star \cup V_{2}^\star)$ and $A_{i}=V_{i}^\star \setminus C$, $i\in [2]$.
Notice that $|A_{1}|+|A_{2}|=n^\star-ξ-|C|-1$, which implies that
\begin{eqnarray}\label{label_klammerausdruckes}
|A_{1}|+|A_{2}|=2n-|C|.
\end{eqnarray}
We set $H$ to be the graph $G^\star \setminus τ^{-1}(V(P))$ and keep in mind that $V(H) = A_1 \cup A_2 \cup (\hat{V}\setminus τ^{-1}(V(P)))$ and $H[A_i], i\in[2]$ is a complete graph.

The fact that $q$ has at least two children in $(T,r)$ implies that $H$ contains a cut-vertex.
Moreover, there exist $H_1, H_2 \in {\sf bc}(H)$ such that for each $i\in[2]$, $H_i$ contains the complete graph $H[A_i]$ as a subgraph.
For each $i\in[2]$, let $T_i$ be the subtree of $T$ induced by the vertices of $τ(V(H_i))$ and $q_i$ be the depth of $T_i$.
Since $H_i$ contains the complete graph $H[A_i]$ as a subgraph, we have that $|A_i|\leq q_i$.
Moreover, the fact that $T$ has depth at most $k' +1$ implies that  $q_{i}\leq k'+1-|V(P)|, i\in[2]$.
Therefore, for each $i\in[2]$, 
\begin{eqnarray}\label{label_procedimiento}
|A_i|\leq k'+1-|V(P)|.
\end{eqnarray}
Also, the fact that $H$ contains a cut-vertex implies that
there is at most one vertex of $\hat{V}$ in $V(H)$.
Thus, $|τ^{-1}(V(P))\cap \hat{V}|\geq \xi$.
We now distinguish two cases, depending whether $\hat{V}\setminus τ^{-1}(V(P))\neq \emptyset$, or not.
\medskip

\noindent{\em Case 1:} $\hat{V}\setminus τ^{-1}(V(P))\neq \emptyset$.
In this case, we have that $|V(P)| = |C|+\xi$ and therefore, by (\ref{label_procedimiento}), $|A_i|\leq 2n-|C|-2k +1, i\in[2]$.
This, together with (\ref{label_klammerausdruckes}), implies that $|A_i|\geq 2k -1, i\in[2]$.
Let $w$ be the (unique) vertex in $\hat{V}\setminus τ^{-1}(V(P))$ and observe that, since $w$ is adjacent to every vertex in $A_1\cup A_2$, $w$ is a cut-vertex of $H$.
This, in turn, implies that there is no edge in $G^\star$ between vertices of $A_1$ and $A_2$.
Thus, in this case, the claim holds for $B_i = A_i, i\in[2]$.\medskip

%

\noindent{\em Case 2:} $\hat{V}\setminus τ^{-1}(V(P))= \emptyset$. Notice that the fact that $\hat{V}\setminus τ^{-1}(V(P))= \emptyset$ implies that $|V(P)| = |C|+\xi+1$.
Therefore, by (\ref{label_procedimiento}), $|A_i|\leq 2n-|C|-2k, i\in[2]$.
This together with (\ref{label_klammerausdruckes}) imply that $|A_i|\geq 2k, i\in[2]$.
Let $z$ be a cut-vertex of $H$ and suppose that $z\in A_1$.
The fact that $z$ is a cut-vertex of $H$ implies that $A_1\setminus z$ and $A_2$ are two subsets of $V_1^\star$ and $V_2^\star$ respectively such that there is no edge, in $G^\star$, between the vertices of $A_1\setminus z$ and $A_2$.
Thus, since $|A_1\setminus \{z\}|\geq 2k-1$ and $|A_2|\geq 2k$, in this case, the claim holds for $B_1= A_1\setminus \{z\}$ and $B_2 =A_2$.\hfill$\diamond$\medskip

Following the Claim, there are sets $B_i\subseteq V_i^\star$ such that
$|B_i|\geq 2k-1, i\in[2]$ and there is no edge between the vertices of $B_1$ and $B_2$ in $G^\star$.
For every $i\in[2]$, since $|B_i|\geq 2k-1$,
there is a set $Q_i$ of at least $k$ vertices such that $Q_i\subseteq V(G)$ or $Q_i\subseteq V'$.
In the former case, we set $W_i = Q_i$, while, in the latter case, we set $W_i= \{v\in V(G)\mid v'\in Q_i\}$.
Therefore, $W_i, i \in[2]$, is a subset of $V(G)$ of size at least $k$.
Since $Q_i \subseteq B_i, i\in[2]$, the fact that there is no edge between the vertices of $B_1$ and $B_2$ in $G^\star$ implies that there is no edge edge between the vertices in $Q_1$ and $Q_2$.
This, in turn, implies that there is no edge in $G^\star$ between $W_1$ and $W_2$,
since otherwise, an edge $uv\in E(G^\star)$ between $W_1$ and $W_2$
would imply the existence of the edges $uv', u'v,$ and $u'v'$ in $G^\star$,
and at least one of them should be between vertices of $Q_1$ and $Q_2$, a contradiction.
Thus, since there is no edge in $G^\star$ between $W_1$ and $W_2$, $W_i, i\in[2]$ induces a complete graph in $Q^\star$, and $G^\star = \overline{G^\bullet}$, it holds that $G[W_1\cup W_2]$ is a complete bipartite graph.
Hence, $(W_1,W_2)$ certifies that $(G,k)$ is a \yes-instance of {\sc BCBS}.

We just proved that {\sc ${\cal E}$-BED} is {\sf NP}-hard.
Our next step is to reduce {\sc ${\cal E}$-BED} to {\sc ${\cal G}$-BED} 
for every non-trivial hereditary class ${\cal G}$. For this consider an instance $(G,k)$
of {\sc ${\cal E}$-BED} and a graph $Z$ as in \autoref{label_caterwauling}.
We construct the graph $G^*$ by considering $|E(G)|$ copies of $Z$
and identify each edge of $G$ with some edge of one of these copies.
Notice that if there is a ${\cal G}$-block tree layout $(F,R,τ)$ of $G^{*}$ of depth at most $k+1$, then there is also one where all vertices of $Z'$ that have not been identified with vertices of $G$ are mapped via $τ$ to subsets of $L(F,R)$.
This implies that  $(G,k)$ is a \yes-instance of {\sc ${\cal E}$-BED} iff 
$(G^*,k)$ is a \yes-instance of {\sc ${\cal G}$-BED}, as required.
 \end{proof}

Notice that
the proof of the above theorem is a (multi) reduction from the problem {\sc Balanced Complete Bipartite Subgraph} (BCBS). It is based on the alternative definition of block elimination distance (\autoref{label_exaggeration}) and has two parts. The first proves the 
{\sf NP}-hardness of {\sc ${\cal E}$-BED}. 
The second is a multi-reduction from  {\sc ${\cal E}$-BED}  to  {\sc ${\cal G}$-BED} 
where the existence of the main gadget is based on the following lemma.

\begin{lemma}[$\star$]
\label{label_caterwauling}
Let ${\cal G}$ be a non-trivial hereditary  class.
Then there exists a graph $Z$ with the following properties: (1) $Z$ is a block graph,
(2) $Z\not\in {\cal B}({\cal G})$ and, (3) $∀ v\in V(Z)$,\ $Z\setminus v\in {\cal B}({\cal G})$.
\end{lemma}

\begin{proof}
Notice that every graph can be seen as an induced subgraph of a block-graph (just add two new universal vertices). This, along with the hereditarity and the   non-triviality of ${\cal G}$ implies that there exists a block graph $H$
that does not belong to ${\cal G}$ and, thus, neither belongs to ${\cal B}({\cal G})$.
Among all induced subgraphs of $H$ that are block graphs and not belonging to  ${\cal B}({\cal G})$, let $Z$ be one with minimum number of vertices.
Clearly, $Z$ satisfies the two first properties. Assume towards a contradiction that 
there is some $v\in V(Z)$ such that $Z\setminus v$ is not biconnected and, moreover, 
$Z\setminus v\not\in {\cal B}({\cal G})$.
It follows that at least one, say $B$, of the blocks of $Z\setminus v$ are not in ${\cal G}$, and thus also not in ${\cal B}({\cal G})$.
Notice that $B$ is a proper induced subgraph of $Z$ (and thus of $H$ as well) that is a block graph and does not belong to ${\cal B}({\cal G})$, a contradiction to the minimality of the choice of ${Z}$.
 \end{proof}

We stress that the  proof of the above lemma is not constructive in the sense that it does not give any way to construct $Z$. However, if the non-trivial and hereditary class ${\cal G}$ is decidable, then $Z$ is effectively computable and this makes the proof of \autoref{label_intermediary} constructive.

\section{Elimination distance to minor-free graph classes}
\label{label_enchantement}

\paragraph{Minors and obstructions.}\label{label_jurisdiction}
The result of the contraction of an edge $e=xy$ in a graph $G$ is the graph  obtained from $G$ after contracting $e$, that is the graph obtained from $G\setminus \{x,y\}$ after introducing a new vertex $v_{xy}$ and edges between $v_{xy}$ and $N_{G}(\{x,y\})\setminus\{x,y\}$. It is denoted by $G/e$. If $H$ can be obtained from some subgraph of $G$ after contracting edges, we say that $H$ is a {\em minor} of $G$ and we denote it by $H\leq G$. 
Given a set ${\cal Q}$ of graphs, we denote by ${\sf excl}({\cal Q})$ the class of all graphs excluding every graph in ${\cal Q}$ as a minor
and by ${\obs}({\cal Q})$ the class of all minor-minimal graphs that do not belong to ${\cal Q}$. Clearly, for every class ${\cal G}$, ${\cal G}={\sf excl}({\sf obs}({\cal G}))$. Also, according to Roberson and Seymour theorem, for every minor-closed class ${\cal G}$, ${\sf obs}({\cal G})$ is finite.
We call a class {\em essential} if it is a finite minor-antichain that is non-empty and does not contain the graph $K_{0}$ or the graph $K_{1}$. Notice that ${\cal G}$ is trivial iff 
${\sf obs}({\cal G})$ is essential.
 We call an essential class ${\cal Z}$ {\em biconnected} if all graphs in ${\sf obs}({\cal Z})$ are block-graphs. 
Given that ${\cal Z}$ is an essential graph class, we define $s({\cal Z})=\max\{|V(G)|\mid G\in {\cal Z}\}$.

%

%

It is easy to verify that  the property of being non-trivial and minor-closed is invariant under both operations ${\cal A}$ and ${\cal B}$.
The most simple example of a non-trivial minor-closed class 
is $\mathcal{E}'=\{K_0,K_{1}\}$ where  ${\sf obs}(\mathcal{E}')=\{K_{1}+K_{1}\}$. Another simple example is 
the class of edgeless graphs ${\cal E}$,  where  ${\sf obs}({\cal E})=\{K_{2}\}$.
Notice that ${\cal B}(\mathcal{E}')={\cal E}\neq \mathcal{E}'$ while ${\cal B}({\cal E})={\cal E}$. In this example 
 ${\cal B}(\mathcal{E}')\neq \mathcal{E}'$. The following easy observation clarifies which classes  are 
 invariants under the  operation ${\cal B}$  and  follows from the fact that 
  for every non-block graph $G$, all graphs in 
${\sf bc}(G)$ are proper minors of $G$.

\begin{observation}
\label{label_constitutional_more}
For every non-trivial minor-closed class ${\cal G}$, ${\cal B}({\cal G})={\cal G}$ iff  ${\sf obs}({\cal G})$ is biconnected.
\end{observation}

\begin{lemma}[$\star$]
\label{label_constitutional}
For every non-trivial minor-closed class ${\cal G}$ and every $k\in \mathbb{N}$, if $Z\in {\sf obs}({\cal G}^{(k)})$, 
then (1) $Z$ is biconnected and (2) every vertex of degree $2$ in $Z$ has adjacent neighbors.
\end{lemma}

\begin{proof}
(1) follows directly from \autoref{label_constitutional_more},  \autoref{label_effectiveness}, and the fact that ${\cal B}({\cal G})={\cal B}({\cal B}({\cal G}))$.

For (2), we consider a biconnected graph $G$ with an edge $e=xy$
and the graph $G^+$ obtained if we remove $e$ from $G$ and 
add a new vertex $v$ adjacent to $x$ and $y$.\smallskip

We claim that if $G$ has a ${\cal G}$-block tree layout of depth $k$, then 
the same holds for $G^+$ as well.
As $G$ is biconnected, we may assume that $(T,\{r\},\tau)$
is a ${\cal G}$-block tree layout  where $T$ is a tree rooted on $r$.
We use the notation $G_{t}:=G[τ^{-1}({\sf d}_{T,\{r\}}(t))]$, $t\in V(T)$.
Let $t\in V(T)$ such that $G_{t}:=G[τ^{-1}({\sf d}_{T,\{r\}}(t))]$
contains the edge $e=xy$ and $e$ is not contained in $G_{t'}$
for some $t'\in {\sf d}_{T,\{r\}}(t)\setminus\{t\}$ (in case $t$ is not a leaf of $T$). Vertex $t$ is unique 
due to condition (2) of ${\cal G}$-block tree layout
and because of the fact that an edge cannot belong to two blocks of a graph.
We update the ${\cal G}$-block tree layout $(T,\{r\},\tau)$ by distinguishing three cases.
If $t$ is a leaf of $T$ and $G_{t}$ is biconnected, then we define $\tau'=\tau\cup\{(v,t)\}$ and $T'=T$.
If $t$ is a leaf of $T$ and  $G_{t}$ is not biconnected, then $G_{t}=(\{x,y\},\{xy\})$ 
and, in this case we define $T'$ by adding a new vertex $t'$ in $T$ and we set $τ'=τ\setminus\{(t,\{x,y\})\}\cup\{(t,\{x,v\}),(t',\{x,v\})\}$. In case $t$ is not a leaf of $T$, then 
one of the  endpoints, say $x$, of $e$ should be mapped, via $\tau$, to $t$.
Then we  define $T'$ by adding a new vertex $t'$ in $T$ adjacent to $t$
and we set $τ'=τ\cup\{(t',\{y,v\})\}$. In any of the above cases $(T',\{r\},\tau')$ is a  ${\cal G}$-block tree layout of $G^+$ of depth $k$. This completes the proof of claim.\smallskip

Suppose now that  there is an obstruction 
$Z$ of ${\cal G}^{(k)}$ that contains some vertex $v$ with two non-adjacent neighbors $x,y$.
As $Z$ is an obstruction, is should be biconnected (by (1)), therefore, from the above claim and \autoref{label_exaggeration}, it follows that the graph $Z'$ obtained by $G$ after contracting the edge $vx$ 
also belongs to ${\cal G}^{(k)}$, a contradiction to the fact that $Z$ is an obstruction.
%
%
 \end{proof}

The next lemma is a direct corollary of \autoref{label_trivialities}. This lemma specifies the structure of the obstructions of the  block closure of every non-trivial minor-closed class and the proof  is postponed  to \autoref{label_sadistically}.   
%
%

\begin{lemma}
\label{label_transportera}
For every essential class ${\cal Z}$, it holds that $s({\sf obs}({\cal B}({\sf excl}({\cal Z}))))\leq 2s  -1$, where $s$ is the maximum number of vertices of a graph in ${\cal Z}$.
\end{lemma}

%
%

An interesting algorithmic consequence of 
\autoref{label_transportera} is the following. The proof is tedious as it recycles standard techniques.
%

\begin{lemma}[$\star$]
\label{label_standardisation}
There is an explicit function $f:\mathbb{N}\to\mathbb{N}$ and an algorithm that, given a finite class ${\cal Z}$, where $s=s({\cal Z})$, a $n$-vertex graph $G$, and an integer $k$, outputs whether ${\sf bed}_{{\sf excl}{(\cal Z)}}(G)\leq k$ in $O(f(s,k)\cdot n^{2})$ time. Moreover, if ${\cal Z}$ contains some planar graph, then the dependence of the running time on $n$ is linear.
\end{lemma}

\begin{proof}
Let $s=s({\cal Z})$, ${\cal G}={\sf excl}({\cal Z})$, and ${\cal Z}^{(k)}={\sf obs}({\cal G}^{(k)})$, for $k\in\mathbb{N}$.
According to the resent result in~\cite{SauST21kapiI}
there is an explicit function $\newfun{sdsfasdfdsf}:\mathbb{N}\to\mathbb{N}$ such that 
 $s({\bf obs}({\cal A}({\cal G}))) ≤ \funref{sdsfasdfdsf}(k,s)$. This, together with \autoref{label_transportera} and \autoref{label_effectiveness}, means that there is an explicit  function $\newfun{sdsfassdfdsf}:\mathbb{N}^{2}\to\mathbb{N}$ such that $s({\sf obs}({\cal G}^{(k)}))\leq \funref{sdsfassdfdsf}(k,s)$.
 The function $\funref{sdsfassdfdsf}$ along with the fact that the problem $\Pi_{\cal G}=\{(G,k)\mid G\in {\cal G}^{(k)}\}$ is decidable (actually, as observed in~\autoref{label_intermediary}, it is in {\sf NP}), implies that  the class ${\cal Z}^{(k)}$ can be constructed by an algorithm whose running time  is  some explicit function, say $\newfun{sdsfassdfsssdsf}:\mathbb{N}^{2}\to\mathbb{N}$, of $k$ and $s$.
Recall now that $G\in {\cal G}^{(k)}$ iff $ ∀Z\in  {\cal Z}^{(k)},$ $Z\nleq G$. Also  because of the algorithmic results in\cite{KawarabayashiKR11thed,RobertsonS95b}, deciding  whether a $z$-vertex graph $Z$ 
is a minor of a $n$-vertex graph $G$ can be done in $O(\newfun{sdsfsassdfsssdsf}(z)\cdot n^{2})$ time where $\funref{sdsfsassdfsssdsf}:\mathbb{N}\to\mathbb{N}$ is some 
explicit function (here $\funref{sdsfsassdfsssdsf}$ is enormous, however, it is indeed explicit -- see~\cite{RobertsonS95b,KawarabayashiW10asho}). This means that, after 
the construction of ${\cal Z}^{(k)}$  one may check whether $G\in {\cal G}^{(k)}$ in $O(\funref{sdsfassdfsssdsf}(k,s)+ \funref{sdsfsassdfsssdsf}(\funref{sdsfassdfdsf}(k,s))\cdot n^{2})$ time.

Suppose now that ${\cal Z}$ contains a planar graph. This, according to~\cite{Chuzhoy15impr}, implies that $\tw({\cal G})= s^{O(1)}$; we use $\tw({\cal G})$ for the maximum treewidth of a graph in ${\cal G}$ (here such a bound will always exist).
It is also easy to see that the treewidth of a non-empty graph is equal to the maximum treewidth of its blocks. This implies that $\tw({\cal  B}({\cal G})) = \tw({\cal G})$. Also, the addition of a vertex does not increase the treewidth of a graph by more than one. This implies that $\tw({\cal A}({\cal G})) \leq   \tw({\cal G})+1$. Given these two observations, and \autoref{label_effectiveness}, we obtain that $\tw({\cal G}^{(k)})=s^{O(1)}+k$. As deciding  whether $\tw(G)\leq q$ can be done in $O(\funref{sdsfasssssssdfdsf}(q)\cdot n)$ steps for some (explicit)
function $\newfun{sdsfasssssssdfdsf}:\mathbb{N}\to\mathbb{N}$ we may assume that $\tw(G)=s^{O(1)}+k$.
Recall that, according to Courcelle's theorem, if ${\cal Q}$ is a class 
for which there is  a formula $φ$ in monadic second order logic where $G\in{\cal Q}$ iff $G\modelsφ$
then there is an explicit function $\newfun{sdssfasssssssdfdsf}:\mathbb{N}^{2}\to\mathbb{N}$ and an algorithm that, given a graph $G$, can check whether $G\in {\cal G}$ in $O(\funref{sdssfasssssssdfdsf}(|φ|,\tw(G))\cdot n)$ time.
Therefore, the second statement of the lemma follows if we give 
a formula $φ_{k}$ such that $G\in {\cal G}^{(k)}$ iff $G\models φ_{k}$.
This follows from the known fact that for every graph $Z$ there is a formula $φ_{Z}$
such that $Z\leq G$ iff $G\modelsφ_{Z}$, therefore,   $G\in {\cal G}^{(k)}$ iff $\forall Z\in  {\cal Z}^{(k)}$, $\neg (G\modelsφ_{Z})$.
%
%
%
 \end{proof}

The next lemma permits us to assume that, in the definition of ${\sf bed}_{\cal G}$, the class ${\cal G}$ can be chosen so that ${\sf obs}({\cal G})$ is biconnected and, moreover,  such a ${\cal G}$ has an explicit obstruction characterization.

\begin{lemma}[$\star$]
\label{label_bouffonneries}
For every  essential class ${\cal Z}$ 
there is a biconnected essential class, in particular the class ${\cal Z'}={\sf obs}({\cal B}({\sf excl}({\cal Z})))$, such that  ${\sf bed}_{{\sf excl}(\cal Z)}$ and ${\sf bed}_{{\sf excl}({\cal Z}')}$ are the same parameter and, moreover, there is an explicit function $\newfun{asdfasdfsdsdfgsdgdfsgdfg}$ such that $s({\cal Z}')\leq \funref{asdfasdfsdsdfgsdgdfsgdfg}(s({\cal Z}))$.
\end{lemma}

\begin{proof}
Let ${\cal G}={\sf excl}({\cal Z})$. From \autoref{label_constitutional_more} and the fact that ${\cal B}({\cal G})={\cal B}({\cal B}({\cal G}))$, ${\sf obs}({\cal B}({\cal G}))$  is  biconnected. The same fact,  together with  \eqref{label_effectiveness},  imply that the parameter ${\sf bed}_{\cal G}$ is the same as ${\sf bed}_{{\cal B}(\cal G)}$.
The bound holds because of \autoref{label_transportera}.
 \end{proof}

\section{Structure of the obstructions for the biconnected closure}
\label{label_sadistically}


\paragraph{Minors.} We start with an alternative definition of the minor relation.
Let $G$ and $H$ be graphs and let $\rho : V(G)\rightarrow V(H)$ be a surjective mapping such that:\vspace{-0mm}
\begin{enumerate}
	\item for every vertex $v\in V(H)$, its codomain $\rho^{-1}(v)$ induces a connected graph $G[\rho^{-1}(v)]$,\vspace{-0mm}
	\item for every edge $\{u,v\}\in E(H)$, the graph $G[\rho^{-1}(u)\cup \rho^{-1}(v)]$ is connected, and\vspace{-0mm}
	\item for every edge $\{u,v\}\in E(G)$, either $\rho(u)=\rho(v)$ or $\{\rho(u), \rho(v)\}\in E(H)$.\vspace{-0mm}
\end{enumerate}
We say that {\em $H$ is a contraction of $G$ (via $\rho$)} and for a vertex $v\in V(H)$ we call the codomain $\rho^{-1}(v)$ the {\em model of $v$} in $G$.
A graph $H$ is a {\em minor of $G$} if there exists a subgraph $M$ of $G$ and a surjective function $\rho: V(M)\to V(H)$ such that $H$ is a contraction of $M,$ via $\rho$.

%
%

%


In this section we prove our next result, which can be seen as the biconnected analog of
\cite[Lemma 5]{BulianD17} where the structure of ${\sf obs}({\cal C}({\sf excl}({\cal Z})))$ is studied. 
The {connected closure} operation in \cite{BulianD17} allows for a shorter less complicated proof, since also the
structure of  graphs in ${\sf obs}({\cal C}({\sf excl}({\cal Z})))$ is simpler. However, in our results, where the deal with the {biconnected closure},
richer structural properties are revealed, resulting also in a more technical proof.

\begin{lemma}[$\star$]\label{label_trivialities}
Let  ${\cal Z}$ be a finite graph class. For every graph $G\in {\sf obs}({\cal B}({\sf excl}({\cal Z})))$
there is  a graph $H\in {\cal Z}$
such that $G$ can be transformed to $H$ after a sequence of at most $|{\sf bc}(H)|-1$ edge deletions and $|{\sf bc}(H)|-1$ edge contractions.
\end{lemma}

\begin{proof}
%
%
Let $G\in {\sf obs}({\cal B}({\sf excl}({\cal Z})))$.
We assume that $|V(G)|\geq 4$, since otherwise the lemma holds trivially.
Since $G\in {\sf obs}({\cal B}({\sf excl}({\cal Z})))$, $G$ is biconnected and also, the fact that $G\notin {\cal B}({\sf excl}({\cal Z}))$ implies that there exists a graph $H\in {\cal Z}$ that is a minor of $G$.
Moreover, since $G$ is a minor-minimal biconnected graph with the latter property, it holds that
\[
\mbox{no proper minor of $G$ is biconnected and contains $H$ as a minor.}\tag{$\star$}
\]
Let $M$ be a (vertex-minimal and, subject to this, edge-minimal) subgraph of $G$ such that 
there exists a surjective function $\rho: V(M)\to V(H)$ such that $H$ is a contraction of $M$ via $\rho$.
As $H$ is a minor of $G$, we know that a pair $(M,\rho)$ as above exists.
We begin with the following claim.
\medskip

\noindent{\em Claim 1:} $G$ can be transformed to $M$ after a sequence of at most $|{\sf bc}(H)|-1$ edge removals.\smallskip
\medskip

\noindent{\em Proof of Claim 1:}
We will prove that $V(M)=V(G)$ and $|E(G)\setminus E(M)|\leq |{\sf bc}(M)|-1$. This, together with the fact that 
$|{\sf bc}(M)|\leq |{\sf bc}(H)|$, will imply Claim 1.
To prove that $V(M)=V(G)$ observe that the existence of a vertex $v\in V(G)\setminus V(M)$ implies that an edge  $e\in E(G)$ incident to $v$ can be either contracted or removed from $G$ while maintaining biconnectivity and the fact that it contains $H$ as a minor, a contradiction to ($\star$).
We now set $E:= E(G)\setminus E(M)$ and we prove that $|E|\leq |{\sf bc}(M)| - 1$ by induction on the number of blocks of $M$.
First, notice that ($\star$) implies that every edge in $E$ is between vertices of different blocks of $M$.
This proves the base case where $|{\sf bc}(M)|=1$.
Suppose that $|{\sf bc}(M)|\geq 2$ and let $B$ be a block of $M$ that contains at most one cut-vertex.
By induction hypothesis, the edges in $\tilde{E} := E \cap E(\cupall ({\sf bc}(G)\setminus \{B\}))$ are at most $|{\sf bc}(M)| - 2$ and by ($\star$), there is at most one edge between the vertices of $B$ and $\cupall( {\sf bc}(G)\setminus \{B\})$.
Therefore, there is at most one edge in $E\setminus \tilde{E}$, which implies that $E\leq |{\sf bc}(M)|-1$.
Hence, Claim 1 follows.\hfill$\diamond$\medskip

We now prove the following. This, combined with Claim 1, completes the proof of the lemma.
\medskip

\noindent{\em Claim 2:} $M$ can be transformed to $H$ after a sequence of at most $|{\sf bc}(H)|-1$ edge contractions.
\medskip

\noindent{\em Proof of Claim 2:}
For every $v\in V(H)$, we set $X_v = \rho^{-1} (v)$.
We will prove that $\sum_{v\in V(H)} |E(G[X_v])| \leq |{\sf bc}(H)| -1$, which implies the above Claim.
We start with a series of observations.\smallskip

\noindent{\em Observation 1:} For every vertex $v\in V(H)$, the graph $G[X_v]$ is a tree.
Indeed, notice that edge-minimality of $M$ implies that $M[X_v]$ is a tree and ($\star$) implies that $E(G[X_v])\subseteq E(M[X_v])$.
\medskip

\noindent{\em Observation 2:} For every vertex $v\in V(H)$ and for every edge $xy\in E(G[X_v])$, $G\setminus \{x,y\}$ is disconnected. Indeed, if there was an edge $e=xy\in E(G[X_v])$ such that $G\setminus \{x,y\}$ is connected, then
$G/e$ would be biconnected, a contradiction to ($\star$).\medskip
%

We now prove that for every vertex $v\in V(H)$ that is not a cut-vertex of $H$, it holds that $|X_v|=1$.
Suppose towards a contradiction that $|X_v|\geq 2$.
By Observation 1, $G[X_v]$ is a tree and therefore, since $|X_v|\geq 2$, there exists an edge $e=xy\in E(G[X_v])$.
The fact that $H\setminus v$ is connected implies that $M\setminus X_v$ and, thus, $G\setminus X_v$ are also connected.
Moreover, due to  ($\star$), every leaf of $G[X_v]$ is adjacent to a vertex of $G\setminus X_v$, which implies that $G\setminus\{x,y\}$ is connected, 
a contradiction to Observation 2.

Next, we argue that the following holds.\medskip

\noindent{\em Subclaim:}
For every cut-vertex $v$ of $H$ it holds that $|E(G[X_v])|\leq  |{\sf cc}(G\setminus X_v)| -1$.
\medskip

\noindent{\em Proof of Subclaim:}
Let $v$ be a cut-vertex of $H$ and let ${\cal Q}=\{Q_1,\ldots, Q_w\}$ be the set ${\sf cc}(G\setminus X_v)$.
By Observation 1, $G[X_v]$ is a tree. For simplicity, we denote by $T$ the graph $G[X_v]$.
For every $i\in [w]$, we set $T_i$ to be the maximum size subtree of $T$ whose leaves are vertices of $N_G (Q_i)$.
We say that an edge $e=\{x,y\}\in E(T)$ is {\em small} if there is an $i\in [w]$ such that $V(T_i)= \{x,y\}$.
Also, given a tree $T$, the {\em internal} edges of $T$ are the ones that are not adjacent to one of its leaves.
%
%
%
We observe that
\begin{enumerate}
\item for every $e_1, e_2\in E(T)$, where $|e_1\cap e_2|=1$, there is an $i\in[w]$ such that $e_1,e_2\in E(T_i)$,
\item for every edge $e\in E(T)$, there are $i,j\in [w]$, where $i\neq j$, such that $e\in E(T_i) \cap E(T_j)$, and
\item every edge of $T$ that is either incident to a leaf of $T$ or an internal edge of some $T_i$, is small.
\end{enumerate}

The fact that $G$ is biconnected implies (1) and (2). To see why (3) holds,
%
%
%
let $e=xy$ be an edge of $T$ such that either one of $x,y$ is a leaf of $T$ or $e$ is an internal edge of some $T_i$, $i\in[w]$ and suppose, towards a contradiction, that $e$ is not small, or, equivalently, for every $j\in [w]$, $|V(T_j)|\geq 3$.
Then, for every vertex $w$ in $V(G)\setminus X_v$ there is a path connecting $w$ with a vertex of $T\setminus \{x,y\}$.
If $e$ is an internal edge of some $T_i$, every pair of vertices in $T\setminus \{x,y\}$ is connected by a path in $G\setminus \{x,y\}$, while if $e$ is incident to a leaf $x$ of $T$, we distinguish two cases:
if $y$ has degree two then $T\setminus \{x,y\}$ is connected, while if $y$ has degree at least three, then, by (1), for every pair $e_1,e_2$ of edges of $T$ incident to $y$, there is an $i\in[w]$ such that $e_1,e_2\in E(T_i)$ and therefore every two vertices in $T\setminus \{x,y\}$ are connected by a path in $G\setminus \{x,y\}$.
Therefore, in all cases, it holds that $G\setminus \{x,y\}$ is connected, a contradiction to Observation 2.
\smallskip

We assume that there is a non-leaf vertex $r\in V(T)$, since otherwise, Subclaim is directly derived from (2).
We consider the rooted tree $(T,r)$.
For every $x\in V(T)$, we consider the subtree $T_x= T[{\sf d}_{T,r} (x)]$ of $T$ and we set ${\sf tc} (x) =  |\{i\in[w]\mid T_i \subseteq T_x \mbox{ and $x$ is a leaf of $T_i$}\}|.$
Observe that $|E(T)| = |E(T_r)|$ and ${\sf tc}(r)\leq w-1$, since, by (1), there is an $i\in [w]$ such that $r$ is an internal vertex of $T_i$.
To conclude the proof of the Subclaim, we prove, by induction on the depth of $T$, that for every $x\in V(T)$,  $|E(T_x)|\leq {\sf tc}(x)$.
Due to (3), for every vertex $x\in V(T)$ that is incident to a leaf of $T$, $|E(T_x)|\leq {\sf tc}(x)$.
Let $T_x$ be a minimum subtree of $T$ whose number of edges is more than ${\sf tc}(x)$ and let $\{y_1,\ldots, y_m\}, m\geq 1$ be the children of $x$ in $(T,r)$.
Since $T_x$ is minimal, for every $i\in[m]$, $|E(T_{y_i})|\leq {\sf tc}(y_i)$.
Let $i\in [m]$ such that $y_i x$ is not small.
Due to (1), there is a $j_i\in [w]$ such that $y_i$ is an internal vertex of $T_{j_i}$ and, due to (3),
$y_i x$ is not an internal edge of $T_{j_i}$, i.e., $T_{j_i}\subseteq T_x$ and $x$ is a leaf of $T_{j_i}$.
Thus, for every $i\in [m]$, either $y_i x$ is small, or  there is a $j_i\in [w]$ such that $T_{j_i}\subseteq T_x$ and $x$ is a leaf of $T_{j_i}$.
Moreover, for every $i,i'\in [m]$, if $i\neq i'$, then $j_i \neq j_{i'}$.
This implies that ${\sf tc}(x)-\sum_{i\in[m]} {\sf tc}(y_i) \geq m$.
Thus, $|E(T_x)| = m + \sum_{i\in [m]} |E(T_{y_i})|  \leq m + \sum_{i\in[m]} {\sf tc}(y_i)\leq {\sf tc}(x)$, a contradiction to our initial assumption that $|E(T_x)|>{\sf tc}(x)$.
Subclaim follows.
\medskip

To conclude the proof of Claim 2, for every cut-vertex $v$ of $H$, we set ${\sf blocks}(H,v)$ to be the blocks of $H$ that contain $v$.
Observe that $|{\sf cc}(G\setminus X_v)|\leq |{\sf blocks}(H,v)|$.
We set ${\sf cv}(H)$ to be the set of cut-vertices of $H$ and we notice that $\sum_{v\in {\sf cv}(H)} |E(G[X_v])| \leq \sum_{v\in {\sf cv}(H)}(|{\sf cc}(G\setminus X_v)| -1)\leq  \sum_{v\in {\sf cv}(H)} (|{\sf blocks}(H,v)| - 1).$
The fact that $\sum_{v\in {\sf cv}(H)} (|{\sf blocks}(H,v)| -1)\leq |{\sf bc}(H)| -1$ implies that
$\sum_{v\in {\sf cv}(H)} |E(G[X_v])| \leq |{\sf bc}(H)| -1$.
The latter together with the fact that  for every vertex $v\in V(H)$ that is not a cut-vertex of $H$, it holds that $|X_v|=1$ completes the proof of Claim 2. \hfill$\diamond$\end{proof}

\begin{figure}[ht]
	\begin{center}
		\includegraphics[width=5.8cm]{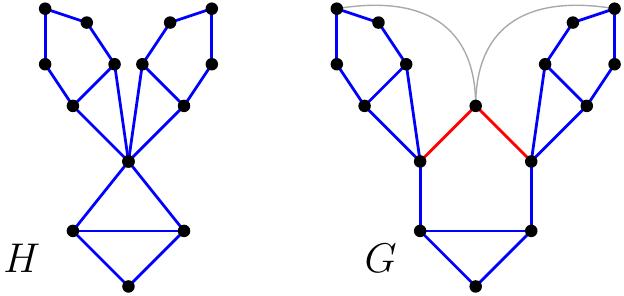}
	\end{center}		\vspace{-0mm}
	\caption{Example of a graph $H$ (on the left) and a graph $G\in {\sf obs}({\cal B}({\sf excl}(H)))$ (on the right) such that $G$ can be transformed to $H$ after exactly $|{\sf bc}(H)|-1$ edge deletions and $|{\sf bc}(H)|-1$ edge contractions.}
\label{label_descubrieras}
\end{figure}

We stress that the bounds on the number of operations in \autoref{label_trivialities} are tight in the sense that, given a graph $H$, there is a graph $G\in {\sf obs}({\cal B}({\sf excl}(\{H\})))$ such that $G$ can be transformed to $H$ after {\sl exactly} $|{\sf bc}(H)|-1$ edge deletions and $|{\sf bc}(H)|-1$ edge contractions.
For example, in \autoref{label_descubrieras}, the graph $H$ on the left has three blocks (i.e., $|{\sf bc}(H)|=3$), the graph $G$ on the right is a graph in  ${\sf obs}({\cal B}({\sf excl}(\{H\})))$, and to transform $G$ to $H$ one has to remove the two grey edges and contract the two red ones.

%

 \section{Outerplanar obstructions for block elimination distance}
 \label{op_rosete}

In this section we study the set ${\sf obs}({\cal G}^{(k)})$ for distinct instantiations of $k$ and ${\cal G}$.
As a warm up,  we prove the following lemma.

\begin{lemma}[$\star$] ${\sf obs}({\cal E}^{(1)})=\{K_{3}\}$  and ${\sf obs}({\cal E}^{(2)})$ consists of the graphs 
${}\atop ^{^{^{\includegraphics[width=4.2cm]{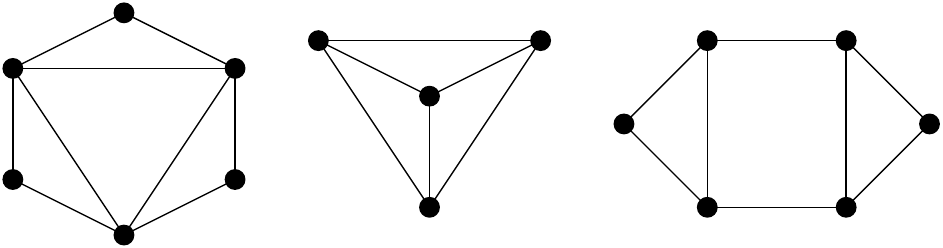}}}}$.
\label{label_unenforceability}
\end{lemma}

\begin{proof}
The fact that ${\sf obs}({\cal E}^{(1)})=\{K_{3}\}$ follows because each 
block of an acyclic graph is a $K_{2}$.
Let us name by $Z_1,Z_2$, and $Z_3$  the three graphs in second part of the statement  from left to right. 
It is easy to verify, by inspection, that $\{Z_1,Z_2,Z_3\}\subseteq {\sf obs}({\cal E}^{(2)})$.
We assume, towards a contradiction,
that $G\in {\sf obs}({\cal E}^{(2)})\setminus \{Z_1,Z_2,Z_3\}$.
%
%
%
Since $G$ is biconnected (\autoref{label_constitutional}.(1)) and $Z_{2}=K_{4}\nleq G$, then
by Dirac's Theorem \cite{Dirac60inabs}, there exist at least two non-adjacent vertices $x$ and $y$ of degree two in $G$. From \autoref{label_constitutional}.(2), each of $x$ and $y$
should belong to a triangle, say $T_{x}$ and $T_{y}$, and $T_{x}$, $T_{y}$ cannot be disjoint as, otherwise, because of its biconnectivity,
$G$ would contain $Z_{3}$ as a minor. Let $w$ be a common neighbor of $x$ and $y$.
%
%
%
%
%
%
%
%
If $G\setminus w$ contains a cycle $C$, then $C$ should intersect both  $T_{x}$ and $T_{y}$, otherwise, again by the biconnectivity, this would imply that  $Z_{3}\leq G$. But then $T_{1}\cup T_{2}\cup C$ contains $Z_{1}$ as a minor, a contradiction. We conclude that $G\setminus w$ is acyclic, therefore it belongs to ${\cal E}^{(2)}$, again a contradiction to the fact that $G\in {\sf obs}({\cal E}^{(2)})\setminus \{Z_1,Z_2,Z_3\}$.
%
%
%
 \end{proof}

\noindent Our objective is to generate obstructions of ${\cal G}^{(k+1)}$ using obstructions of ${\cal G}^{(k)}$. 
For this, we  define the following two operations. (See also \autoref{label_minuet}.)

\begin{itemize}
\item {\sl Parallel join}: Let $G_{1}$ and $G_{2}$ be graphs and let $v_{1}^{i},v_{2}^{i}\in V(G_{i})$, $i\in [2]$.
We denote by
$||(G_{1},v_{1}^{1},v_{2}^{1},G_{2},v_{1}^{2},v_{2}^{2})$ the graph obtained from the disjoint 
union of $G_{1}$ and $G_{2}$ after we add the edges $\{v_{i}^{1},v_{i}^{2}\}, i\in[2]$ and we call it the {\em parallel join} of $G_{1}$ and $G_{2}$ on $(v^{1}_{1},v^{1}_{2})$ and $(v^{2}_{1},v^{2}_{2})$.
\item {\em Triangular gluing}: Let $G_{1},G_{2},$ and $G_{3}$ be graphs
and let $v_{1}^{i},v_{2}^{i}\in V(G_{i})$, $i\in[3]$.
We denote by $\triangle(G_{1},v_{1}^{1},v_{2}^{1},G_{2},v_{1}^{2},v_{2}^{2},G_{3},v_{1}^{3},v_{2}^{3})$ the graph obtained from the disjoint union of $G_{1}$, $G_{2}$, and 
$G_{3}$ after we identify the pairs $v_{2}^{1}$ and $v_{1}^{2}$,
$v_{2}^{2}$ and $v_{1}^{3}$, and $v_{2}^{3}$ and $v_{1}^{1}$. We call this graph the {\em triangular gluing} of $G_{1}$, $G_{2}$, and $G_{3}$ on $(v^{1}_{1},v^{1}_{2})$, 
$(v^{2}_{1},v^{2}_{2})$, and $(v^{3}_{1},v^{3}_{2})$.
\end{itemize}

\begin{figure}[!h]
\centering
\begin{subfigure}{0.4\textwidth}
\centering
\begin{tikzpicture}[scale=0.8]
\begin{scope}[xshift=-1.5cm]
\node[v:main] (l1) at (35:1){};
\node[v:ghost] () at (35:0.6){$v^{1}_{1}$};
\node[v:main] (l2) at (-35:1){};
\node[v:ghost] () at (-35:0.6){$v^{1}_{2}$};
\end{scope}

\begin{scope}[xshift=1.5cm]
\node[v:main] (r1) at (145:1){};
\node[v:ghost] at (145:0.6){$v^{2}_{1}$};
\node[v:main] (r2) at (-145:1){};
\node[v:ghost] at (-145:0.6){$v^{2}_{2}$};
\end{scope}

\node () at (-1.9,0.3){$G_{1}$};

\node () at (1.9,0.3){$G_{2}$};

\begin{pgfonlayer}{background}
\draw[e:main] (l1) to (r1){};
\draw[e:main] (l2) to (r2){};

\draw[thick,fill=blue!40] (1.5,0) circle (1cm);
\draw[thick,fill=blue!40] (-1.5,0) circle (1cm);

\end{pgfonlayer}

\end{tikzpicture}
\end{subfigure}
\begin{subfigure}{0.4\textwidth}
\centering
\begin{tikzpicture}[scale=0.8]
\pgfdeclarelayer{background}
\pgfdeclarelayer{foreground}
\pgfsetlayers{background,main,foreground}

\node[v:main] () at (0,2.5){};
\node[v:main] () at (-1.5,0){};
\node[v:main] () at (1.5,0){};

\node[v:ghost] () at (-0.75,1.3){$G_{1}$};
\node[v:ghost] () at (0.75,1.3){$G_{3}$};
\node[v:ghost] () at (0,0){$G_{2}$};

\node[v:ghost] () at (2.45,0){$v^{2}_{2}=v^{3}_{1}$};
\node[v:ghost] () at (0,2.9){$v^{1}_{1}=v^{3}_{2}$};
\node[v:ghost] () at (-2.45,0){$v^{1}_{2}=v^{2}_{1}$};
\begin{pgfonlayer}{background}

\draw[e:main, bend left,fill=blue!40] (-1.5,0) to (0,2.5) to (-1.5,0){};
\draw[e:main, bend left,fill=blue!40] (1.5,0) to (0,2.5) to (1.5,0){};

\draw[e:main, bend left,fill=blue!40] (-1.5,0) to (1.5,0) to (-1.5,0){};

\end{pgfonlayer}
\end{tikzpicture}
\end{subfigure}
\vspace{-2mm}
\caption{On the left side we see the graph resulting from the parallel join of the graphs $G_{1}$ and $G_{2}$ on $(v^{1}_{1},v^{1}_{2})$ and $(v^{2}_{1},v^{2}_{2})$. 
On the right side we see the graph resulting from the triangular gluing of the graphs $G_{1}$, $G_{2}$, and $G_{3}$ on $(v^{1}_{1},v^{1}_{2})$, $(v^{2}_{1},v^{2}_{2})$, and $(v^{3}_{1},v^{3}_{2})$.}
\label{label_minuet}
\end{figure}
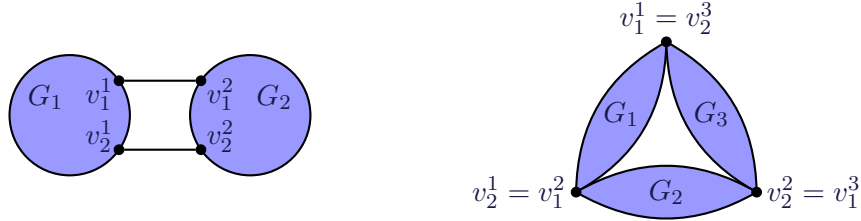

By the above constructions we can make the following observation.
\begin{observation}
Let $G_{1}$, $G_{2}$, and $G_{3}$ be graphs and $v_{1}^{i},v_{2}^{i}\in V(G_{i})$, $i\in[3]$. If $G_{1}$, $G_{2}$, and $G_{3}$ are biconnected then so are the graphs 
$||(G_{1},v_{1}^{1},v_{2}^{1},G_{2},v_{1}^{2},v_{2}^{2})$ and $\triangle(G_{1},v_{1}^{1},v_{2}^{1},G_{2},v_{1}^{2},v_{2}^{2},G_{3},v_{1}^{3},v_{2}^{3})$.
\end{observation}

\begin{lemma}[$\star$]
\label{label_conflictivitat}
Let ${\cal G}$ be a non-trivial and minor-closed class and $k\in\mathbb{N}$.
If $\mathsf{bed}_{\mathcal{G}}(G_{i})\geq k+1$, $i\in [2]$,  $v_{1}^{i},v_{2}^{i}\in V(G_{i}), i\in[2]$, and the graph $G=||(G_{1},v_{1}^{1},v_{2}^{1},G_{2},v_{1}^{2},v_{2}^{2})$ is biconnected, then $\mathsf{bed}_{\mathcal{G}}(G)\geq k+2$.
Moreover, under the assumption that  either ${\cal G}\neq{\cal E}$ or $k\geq 1$, the following holds:   
if $G_{1},G_{2}\in {\sf obs}({\cal G}^{(k)})$ and  $v_{1}^{i},v_{2}^{i}\in V(G_{i}), i\in[2]$, then $G\in {\sf obs}({\cal G}^{(k+1)})$.
\end{lemma}

\begin{proof}[Proof of \autoref{label_conflictivitat}]
Let $G_{1}$ and $G_{2}$ be two graphs such that $\mathsf{bed}_{\mathcal{G}}(G_{i})\geq k+1$, $i\in [2]$, and $v_{1}^{i},v_{2}^{i}\in V(G_{i}), i\in[2]$. 
We denote by $G$ the graph $||(G_{1},v_{1}^{1},v_{2}^{1},G_{2},v_{1}^{2},v_{2}^{2})$. First, we prove that 
$\mathsf{bed}_{\mathcal{G}}(G)\geq k+2$. Indeed, assume to the contrary, that 
$\mathsf{bed}_{\mathcal{G}}(G)\leq k+1$. 
Hence, there exists a vertex $v\in V(G)$ such that $k+1\geq 1+\mathsf{bed}_{\mathcal{G}}(G\setminus v)$ 
and therefore $\mathsf{bed}_{\mathcal{G}}(G\setminus v)\leq k$. Notice that $v$ is either a vertex of $V(G_{1})$ or a vertex of $V(G_{2})$. 
Without loss of generality, let $v\in V(G_{2})$. Then $G_{1}\subseteq G\setminus v$ and thus, 
$\mathsf{bed}_{\mathcal{G}}(G_{1})\leq \mathsf{bed}_{\mathcal{G}}(G\setminus v)\leq k$. 
A contradiction to the hypothesis that $\mathsf{bed}_{\mathcal{G}}(G_{1})\geq k+1$.

We now prove that, if $G_{1},G_{2}\in {\sf obs}({\cal G}^{(k)})$ then $G\in {\sf obs}({\cal G}^{(k+1)})$.
From~\autoref{label_constitutional}.(1) and the first part of the lemma it follows that $\mathsf{bed}_{\mathcal{G}}(G)\geq k+2$. Therefore, in order to prove that $G\in \mathsf{obs}(\mathcal{G}^{(k+1)})$
it is enough to prove that, any vertex/edge deletion or edge contraction on $G$ decreases the parameter by 1. 
In particular, notice that it suffices to show that any edge deletion or edge contraction decreases the parameter by 1.
Let $e\in E(G)$. Then either $e\in E(G_{1})$ or $e\in E(G_{2})$ or $e=\{v_{i}^{1},v_{i}^{2}\}$ for some $i\in[2]$. 
Let us first consider the case where $e\in E(G_{1})$ (the case where $e\in E(G_{2})$ is symmetrical). 
We will prove that $\mathsf{bed}_{\mathcal{G}}(G\setminus e)\leq k+1$ and $\mathsf{bed}_{\mathcal{G}}(G/e)\leq k+1$.
Let $G'=G\setminus e$ and let $v=v_{1}^{2}\in V(G_{2})$. Then $\mathsf{bed}_{\mathcal{G}}(G')\leq 1+\mathsf{bed}_{\mathcal{G}}(G'\setminus v)$. 
Let $H$ be a biconnected component of $G'\setminus v$. Notice that either $H\subseteq G_{2}\setminus v$ or $H\subseteq G_{1}\setminus e$. 
Since $G_{1}$ and $G_{2}$ belong to $\mathsf{obs}(\mathcal{G}^{(k)})$, $\mathsf{bed}_{\mathcal{G}}(G_{1}\setminus e)\leq k$ and $\mathsf{bed}_{\mathcal{G}}(G_{2}\setminus v)\leq k$. 
Therefore $\mathsf{bed}_{\mathcal{G}}(H)\leq k$. 
This implies that $\mathsf{bed}_{\mathcal{G}}(G'\setminus v)\leq k$ and thus $\mathsf{bed}_{\mathcal{G}}(G')\leq k+1$. 
Let now $G'=G/e$ and observe that the proof above also argues that $\mathsf{bed}_{\mathcal{G}}(G')\leq k+1$, after replacing $G_{1}\setminus e$ by $G_{1}/e$.

Finally, we consider the case where $e=v_{1}^{1}v_{1}^{2}$ (the case where $e=v_{2}^{1}v_{2}^{2}$ is symmetrical).
Let $G'=G\setminus e$. Notice that $G'$ is not biconnected and, moreover, its blocks are the graphs $G_{1}$ and $G_{2}$, and the bridge $B$ consisting of the edge $v_{1}^{1}v_{1}^{2}$. Recall that $G_{1},G_{2}\in\mathsf{obs}(\mathcal{G}^{(k)})$ and thus $\mathsf{bed}_{\mathcal{G}}(G_{i})=k+1$, $i\in [2]$.
Therefore, $\mathsf{bed}_{\mathcal{G}}(G')=\max\{\mathsf{bed}_{\mathcal{G}}(G_{1}),\mathsf{bed}_{\mathcal{G}}(G_{2}), {\sf bed}_{\cal G}(B)\}=k+1$.
Let now $G'=G/e$ and let $v=v_{2}^{1}$. As before, notice that the graph $G'\setminus v$ is not biconnected. Moreover, if $H$ is a block of $G'\setminus v$ then either $H=G_{2}$ or $H\subseteq G_{1}\setminus v_{2}^{1}$. Again, we obtain that $\mathsf{bed}_{\mathcal{G}}(H)\leq k+1$. This concludes the proof of the lemma.
 \end{proof}

\begin{lemma}[$\star$]
\label{label_forsakenness}
Let ${\cal G}$ be a non-trivial and minor-closed class and $k\in\mathbb{N}$. 
If $\mathsf{bed}_{\mathcal{G}}(G_{i})\geq k+1$, $i\in [3]$,  $v_{1}^{i},v_{2}^{i}\in V(G_{i}), i\in[3]$, and the graph $G=\triangle(G_{1},v_{1}^{1},v_{2}^{1},G_{2},v_{1}^{2},v_{2}^{2},G_{3},v_{1}^{3},v_{2}^{3})$ is biconnected,  then $\mathsf{bed}_{\mathcal{G}}(G)\geq k+2$.
Moreover, if  $G_{1},G_{2},G_{3}\in {\sf obs}({\cal G}^{(k)})$ and  $v_{1}^{i},v_{2}^{i}\in V(G_{i}), i\in[3]$, then 
$G  \in {\sf obs}({\cal G}^{(k+1)})$.
\end{lemma}

\begin{proof}[Proof of \autoref{label_forsakenness}]
Let $G_{1}$, $G_{2}$, and $G_{3}$ be three graphs such that 
$\mathsf{bed}_{\mathcal{G}}(G_{i})\geq k+1$, $i\in [3]$.
We denote by $G$ the graph $\triangle(G_{1},v_{1}^{1},v_{2}^{1},G_{2},v_{1}^{2},v_{2}^{2},G_{3},v_{1}^{3},v_{2}^{3})$ and prove that 
$\mathsf{bed}_{\mathcal{G}}(G)\geq k+2$.
Indeed, assume to the contrary, that 
$\mathsf{bed}_{\mathcal{G}}(G)\leq k+1$. 
Hence, there exists a vertex $v\in V(G)$ such that $k+1\geq 1+\mathsf{bed}_{\mathcal{G}}(G\setminus v)$ 
and therefore $\mathsf{bed}_{\mathcal{G}}(G\setminus v)\leq k$. Notice that there exists $i\in [3]$ such that $v\notin V(G_{i})$. 
Without loss of generality, let $v\notin V(G_{2})$. Then $G_{1}\subseteq G\setminus v$ and thus, 
$\mathsf{bed}_{\mathcal{G}}(G_{1})\leq \mathsf{bed}_{\mathcal{G}}(G\setminus v)\leq k$. 
A contradiction to the hypothesis that $\mathsf{bed}_{\mathcal{G}}(G_{1})\geq k+1$.

We now prove that if  $G_{1},G_{2},G_{3}\in {\sf obs}({\cal G}^{(k)})$ then 
$G  \in {\sf obs}({\cal G}^{(k+1)})$.
The first part of the lemma, combined with~\autoref{label_constitutional}.(1), proves that $\mathsf{bed}_{\mathcal{G}}(G)\geq k+2$. Hence, it is enough to prove that
 any vertex/edge deletion or edge contraction on $G$ decreases the parameter by 1. 
In particular, notice that it suffices to show that any edge deletion or edge contraction decreases the parameter by 1.
Let $e\in E(G)$. Then $e\in E(G_{i})$ for some $i\in [3]$. Without loss of generality, we assume that $e\in E(G_{3})$.  Recall that the vertices $v_{2}^{1}\in V(G_{1})$ 
and $v_{1}^{2}\in V(G_{2})$ have been identified in $G$. We denote them by $v$.

Let $G'=G\setminus e$. We will prove that $\mathsf{bed}_{\mathcal{G}}(G')\leq k+1$. 
We distinguish two cases according to whether $G'$ is biconnected. 
Let us first assume that $G'$ is not biconnected. It follows that if $H$ is a block of $G'$ then either 
$H\subseteq G_{i}$ for some $i\in [2]$, or $H\subseteq G_{3}\setminus e\subseteq G_{3}$.
Since $G_{i}\in \mathsf{obs}(\mathcal{G}^{(k)})$, $i\in [3]$, it holds that $\mathsf{bed}_{\mathcal{G}}(G_{i})=k+1$ and hence $\mathsf{bed}_{\mathcal{G}}(H)\leq k+1$, for every block $H$ of $G'$. 
Moreover, from the definition, $\mathsf{bed}_{\mathcal{G}}(G')=\max\{\mathsf{bed}_{\mathcal{G}}(H) \mid H~\text{is a block of}~G'\}$. 
Thus, $\mathsf{bed}_{\mathcal{G}}(G')\leq k+1$.

Let us now assume that $G'$ is biconnected. Then, $\mathsf{bed}_{\mathcal{G}}(G')\leq 1+\mathsf{bed}_{\mathcal{G}}(G'\setminus v)$.
Observe that $G'\setminus v$ is not biconnected. Moreover, if $H$ is a block of $G'\setminus v$ then either $H\subseteq G_{i}\setminus v$, for some $i\in [2]$, or 
$H\subseteq G_{3}\setminus e$. Observe that since $G_{i}\in \mathsf{obs}(\mathcal{G}^{(k)})$, $i\in [3]$ it holds that $\mathsf{bed}_{\mathcal{G}}(G_{i}\setminus v)\leq k$, $i\in [2]$,
and $\mathsf{bed}_{\mathcal{G}}(G_{3}\setminus e)\leq k$. Thus, $\mathsf{bed}_{\mathcal{G}}(H)\leq k$, for any block $H$ of $G'\setminus v$. Therefore,
$\mathsf{bed}_{\mathcal{G}}(G')\leq 1+\mathsf{bed}_{\mathcal{G}}(G'\setminus v)\leq 1+\max\{\mathsf{bed}_{\mathcal{G}}(H)\mid H~\text{is a block of}~G'\setminus v\}
\leq 1+k$. This concludes the proof that the removal of any edge from $G$ decreases the parameter by 1.

To conclude, observe that the above arguments hold for the case that we consider edge contractions instead of deletions, 
if we replace $G'\setminus e$ by $G'/e$ and $G_{3}\setminus e$ by $G_{3}/e$. 
 \end{proof}
 
\autoref{label_conflictivitat} and \autoref{label_forsakenness} imply that the set $\bigcup_{i\geq 0}\mathsf{obs}(\mathcal{G}^{(i)})$
is closed under the parallel join and the triangular gluing operations.

The following is also a consequence of \autoref{label_conflictivitat} and \autoref{label_forsakenness}.
\begin{lemma}\label{label_asturias}
Let $\mathcal{G}$ be a non-trivial and minor-closed graph class and $k\in \mathbb{N}$. Let also $G$, $G_{1}$, $G_{2}$, and $G_{3}$ be graphs and 
$v^{i}_{1},v^{i}_{2}\in V(G_{i})$, $i\in [3]$. If $G=\triangle(G_{1},v_{1}^{1},v_{2}^{1},G_{2},v_{1}^{2},v_{2}^{2},G_{3},v_{1}^{3},v_{2}^{3})$ (or
$G=||(G_{1},v_{1}^{1},v_{2}^{1},G_{2},v_{1}^{2},v_{2}^{2})$) and $G\in \mathsf{obs}(\mathcal{G}^{(k+1)})$ then $G_{i}\in \mathsf{obs}(\mathcal{G}^{(k)})$ for all $i\in [3]$ (or $i\in [2]$). 
\end{lemma}

We denote by ${\cal O}$  the 
class of all outerplanar graphs.
We claim that ${\cal O}\cap \bigcup_{i\geq 1}\mathsf{obs}(\mathcal{G}^{(i)})$ is complete 
under these two operations. In particular we prove the following:

\begin{lemma}[$\star$]
\label{label_desencantado}Let ${\cal G}$ be a non-trivial and minor-closed class.
For every $k\in\mathbb{N}_{\geq 1}$ and for  every graph $G\in\mathsf{obs}(\mathcal{G}^{(k+1)})\cap {\cal O}$,
there are 
\begin{itemize}
\item either two graphs $G_{1}$ and $G_{2}$  of ${\sf obs}({\cal G}^{(k)})\cap {\cal O}$
and $v_{1}^{i},v_{2}^{i}\in V(G_{i})$, $i\in[2]$,
such that $G=||(G_{1},v_{1}^{1},v_{2}^{1},G_{2},v_{1}^{2},v_{2}^{2})$ or 

\item  three graphs $G_{1}, G_{2}$ and $G_{3}$
 of ${\sf obs}({\cal G}^{(k)})\cap {\cal O}$ and $v_{1}^{i},v_{2}^{i}\in V(G_{i})$, $i\in[3]$,
such that $G=\triangle (G_{1},v_{1}^{1},v_{2}^{1},G_{2},v_{1}^{2},v_{2}^{2},G_{2},v_{1}^{3},v_{2}^{3})$.
\end{itemize} 
\end{lemma}

Before we begin the proof of \autoref{label_desencantado} we need a series of  definitions.

Let $A$ be a subset of the plane $\mathbb{R}^2$. We define ${\bf int}(A)$ to be the interior of $A$, $\mathbf{cl}(A)$ its closure and ${\bf bd}(A) = \mathbf{cl}(A) \setminus {\bf int}(A)$ its border.
Given a plane graph $\Gamma$ (that is a graph embedded in $\mathbb{R}^2$), we denote its {\em faces} by $F(\Gamma)$, that is, $F(\Gamma)$
is the set of the connected components of $\mathbb{R}^{2}\setminus \Gamma$  (in the operation $\mathbb{R}^{2}\setminus \Gamma$ we treat 
$\Gamma$ as the set of points of $\mathbb{R}^2$ corresponding to its vertices and its edges). Observe that $\mathbb{R}^{2}\setminus \Gamma$ contains exactly one unbounded face, 
which we call {\em outer} face and denote it by $f_{o}$. All other faces are called {\em inner} faces.
For every $f \in F(\Gamma)$ we denote by 
$B_{\Gamma}(f)$ the graph induced by the vertices and edges of $\Gamma$ whose embeddings are subsets of ${\bf bd}(f)$ and we call it the {\em boundary} of $f$.

Let $\Gamma$ be a fixed outerplanar embedding of an outerplanar graph $G$. Thus, all vertices of $G$ belong to $B_{\Gamma}(f_{o})$.
Let $\Gamma^{*}$ be the graph obtained from $\Gamma$ in the following way. Its vertex set is
the set 
$$V(\Gamma^{*})=\{v_{f}\mid f \in F(\Gamma)\setminus f_{o}\}\cup\{v_{e}\mid e\in E(B_{\Gamma}(f_{o}))\}.$$
That is, $\Gamma^{*}$ contains a vertex for every inner face of $\Gamma$ and a vertex for every edge of $\Gamma$ that belongs to the graph induced by the boundary of its outer face.
Moreover, its edge set is 
$$E(\Gamma^{*})=\{v_{f_{1}}v_{f_{2}}\mid f_{1}\neq f_{2}\text{~and~}E(B_{\Gamma}(f_{1}))\cap E(B_{\Gamma}(f_{2}))\neq\emptyset\}\cup \{v_{f}v_{e}\mid  e\in E(B_{\Gamma}(f))\},$$
that is, two vertices are connected by an edge if one of the two following holds: Either both vertices correspond to distinct inner faces whose boundary graphs share an edge or one of the vertices corresponds to an inner face that shares an edge with the outer face and the other vertex corresponds to that edge.
We call an edge of $\Gamma^{*}$ that contains a vertex $v_{e}$, for some $e\in E(\Gamma)$ (and in particular $e\in B_{\Gamma}(f_{o})$), {\em marginal}. Otherwise, we call it {\em internal}. 
Finally, we call $\Gamma^{*}$ the {\em weak dual} of $\Gamma$.
The parameter $\mathsf{bed}_{\mathcal{G}}$ on embedded graphs $\Gamma$ is defined as the parameter $\mathsf{bed}_{\mathcal{G}}$ on the underlying combinatorial graph $G$.

The following observation is folklore and we skip its proof.
\begin{observation}\label{label_entredosaguas}
If $\Gamma$ is an outerplanar embedding of a graph then $\Gamma^{*}$ is a tree. Moreover, all of its leaves belong to marginal edges and each marginal edge contains a leaf of $T$.
\end{observation}

Let $e=v_{f_{1}}v_{f_{2}}$ be an internal edge of $\Gamma^{*}$. Let $e_{f_{1},f_{2}}$ denote the edge in $E(B_{\Gamma}(f_{1}))\cap E(B_{\Gamma}(f_{2}))$. 
By construction, $e_{f_{1},f_{2}}\notin E(B_{\Gamma}(f_{o}))$.
This implies that the endpoints of $e_{f_{1},f_{2}}$ form a separator of $\Gamma$. Let $\Gamma_{e,f_{1}}',\Gamma_{e,f_{2}}'$ be the connected components of $\Gamma\setminus e_{f_{1},f_{2}}$ 
such that $\Gamma_{e,f_{i}}'\cap B_{\Gamma}(f_{i})\neq\emptyset$ (here, we interpret $e_{f_{1},f_{2}}$
as the vertex set containing the endpoints of the edge $e_{f_{1},f_{2}}$).
We denote by $\Gamma_{e,f_{i}}$ the embedded graph induced by $V(\Gamma_{f_{i}}')\cup e_{f_{1},f_{2}}$ (where, again, we interpret $e_{f_{1},f_{2}}$
as the vertex set containing the endpoints of the edge $e_{f_{1},f_{2}}$).
\smallskip

We now proceed with the proof of \autoref{label_desencantado}.

\begin{proof}[Proof of \autoref{label_desencantado}]
Let $G$ be an outerplanar graph such that $G\in \mathsf{obs}(\mathcal{G}^{(k+1)})$. Since $G\in \mathsf{obs}(\mathcal{G}^{(k+1)})$, from~\autoref{label_constitutional}.(1), $G$ is biconnected and thus has a unique 
outerplanar embedding on the plane. We denote its unique embedding by $\Gamma$.
Let $\Gamma^{*}$ be the weak dual of $\Gamma$. From \autoref{label_entredosaguas}, $\Gamma^{*}$ is a tree. 
We orient the edges of $\Gamma^{*}$ in the following way. The marginal edges are oriented away from their incident leaf.
Let $e=v_{f_{1}}v_{f_{2}}$ be an internal edge of $\Gamma^{*}$. If $\mathsf{bed}_{\mathcal{G}}(\Gamma_{f_{1}})>\mathsf{bed}_{\mathcal{G}}(\Gamma_{f_{2}})$, 
we orient the edge towards $v_{f_{1}}$. 
Symmetrically, if $\mathsf{bed}_{\mathcal{G}}(\Gamma_{f_{2}})>\mathsf{bed}_{\mathcal{G}}(\Gamma_{f_{1}})$, we orient the edge towards $v_{f_{2}}$. We call these edges unidirectional.
Otherwise, we orient the edge in both directions and call it bidirectional.

We will use the oriented tree to prove that $G$ can be decomposed in one of the two ways stated in the lemma. Towards this, we prove the following claims.\medskip

\noindent{\em Claim 1:} For every internal edge $e=v_{f_{1}}v_{f_{2}}$, it holds that $\mathsf{bed}_{\mathcal{G}}(\Gamma_{e,f_{i}})\leq k+1$, $i\in [2]$. Moreover,
if $e$ is bidirectional then $\mathsf{bed}_{\mathcal{G}}(\Gamma_{e,f_{1}})=\mathsf{bed}_{\mathcal{G}}(\Gamma_{e,f_{2}})= k+1$
and if $e$ is unidirectional oriented from $v_{f_{1}}$ to $v_{f_{2}}$ then $\mathsf{bed}_{\mathcal{G}}(\Gamma_{e,f_{1}})\leq k$.\medskip

\noindent{\em Proof of Claim 1:} Indeed, both statements follow from the facts that $\Gamma_{f_{i}}$ is a proper subgraph of $G$ and $G\in\mathsf{obs}(\mathcal{G}^{(k+1)})$.
The first statement is straightforward. For the second statement, let us assume first that $\mathsf{bed}_{\mathcal{G}}(\Gamma_{e,f_{1}})=\mathsf{bed}_{\mathcal{G}}(\Gamma_{e,f_{2}})\leq k$.
Let also $uv$ be the common edge of the graphs $\Gamma_{e,f_{1}}$ and $\Gamma_{e,f_{2}}$.
Observe then that if $H$ is a block of the graph $G\setminus u$ then $H$ is a subgraph of one of the two graphs $\Gamma_{e,f_{i}}\setminus u$, $i\in [2]$. This implies that
$\mathsf{bed}_{\mathcal{G}}(H)\leq \mathsf{bed}_{\mathcal{G}}(\Gamma_{e,f_{i}}\setminus u) \leq \mathsf{bed}_{\mathcal{G}}(\Gamma_{e,f_{i}}) \leq k$, for some $i\in [2]$.
Hence, by definition, we get that $\mathsf{bed}_{\mathcal{G}}(G)\leq k+1$, a contradiction to the hypothesis that $G\in\mathsf{obs}(\mathcal{G}^{(k+1)})$.\hfill$\diamond$\medskip

\noindent{\em Claim 2:} There does not exists a vertex of $\Gamma^{*}$ incident to two distinct edges, such that both of them are oriented away from it and
at least one of them is unidirectional.\medskip

\noindent{\em Proof of Claim 2:} Indeed, let us assume that such a vertex exists and let $e_{1}$ and $e_{2}$ be two edges
oriented away from it and without loss of generality let $e_{1}$ be the edge that is unidirectional. 
Notice that the assumed vertex is an internal vertex of the tree $\Gamma^{*}$. We will denote it by $v_{f}$. Moreover, by definition of the orientations, 
the two distinct endpoints of $e_{1}$ and $e_{2}$ are also internal vertices of the tree. We denote them by $v_{f_{1}}$ and $v_{f_{2}}$, respectively. 
Notice that $\Gamma_{e_{2},f_{2}}$ is a subgraph of $\Gamma_{e_{1},f}$ and that  $\Gamma_{e_{1},f_{1}}$ is a subgraph of $\Gamma_{e_{2},f}$. Therefore,
$\mathsf{bed}_{\mathcal{G}}(\Gamma_{e_{2},f_{2}})\leq \mathsf{bed}_{\mathcal{G}}(\Gamma_{e_{1},f})$ and $\mathsf{bed}_{\mathcal{G}}(\Gamma_{e_{1},f_{1}}) \leq \mathsf{bed}_{\mathcal{G}}(\Gamma_{e_{2},f})$. Moreover, $\mathsf{bed}_{\mathcal{G}}(\Gamma_{e_{1},f}) < \mathsf{bed}_{\mathcal{G}}(\Gamma_{e_{1},f_{1}})$, since $e_{1}$ is uniquely oriented towards $f_{1}$.
This implies that $\mathsf{bed}_{\mathcal{G}}(\Gamma_{e_{2},f_{2}})< \mathsf{bed}_{\mathcal{G}}(\Gamma_{e_{2},f})$, a contradiction to the assumption that $e_{2}$ is oriented towards $f_{2}$.
This completes the proof of the claim.\hfill$\diamond$\medskip

\noindent{\em Claim 3:} There exists a bidirectional edge in $\Gamma^{*}$ (by construction, this edge is internal). \medskip

\noindent{\em Proof of Claim 3}: Assume, towards a contradiction, that all edges of $\Gamma^{*}$ have a unique direction. Then Claim 2 implies that there exists a unique vertex in 
$\Gamma^{*}$ that is a sink, that is, all edges are oriented towards it. It follows that this vertex is internal. Let us denote it by $v_{f}$. Let us denote by $e_{i}$ denote the internal edges incident 
to $v_{f}$ and $v_{f_{i}}$ denote their other endpoints, $i\in [x]$, where by $x$ we denote the number of internal edges incident to $v_{f}$. Finally, let $u \in B_{\Gamma^{*}}(f)$, 
and notice that if $H$ is a block of $\Gamma^{*}\setminus u$ then $H\subseteq \Gamma_{e_{i},f_{i}}\setminus u\subseteq \Gamma_{e_{i},f_{i}}$ for some $i\in [x]$.
From Claim 1, we obtain that $\mathsf{bed}_{\mathcal{G}}(H)\leq \mathsf{bed}_{\mathcal{G}}(\Gamma_{e_{i},f_{i}})\leq k$, for some $i\in [x]$. From the definition of 
$\mathsf{bed}_{\mathcal{G}}$ we obtain that $\mathsf{bed}_{\mathcal{G}}(\Gamma)\leq k+1$, a contradiction to the hypothesis that $G\in\mathsf{obs}(\mathcal{G}^{(k+1)})$.\hfill$\diamond$\medskip

\begin{figure}[!h]
	\centering
	\begin{subfigure}{0.3\textwidth}
	\centering
\begin{tikzpicture}[scale=0.5]

\pgfdeclarelayer{background}
\pgfdeclarelayer{foreground}
\pgfsetlayers{background,main,foreground}

\node[v:main] () at (0,6){};
\node[v:ghost] () at (0,6.4){$v$};

\node[v:main] () at (0,0){};
\node[v:ghost] () at (0,-0.4){$u$};

\node[v:main] () at (-2,5){};
\node[v:ghost] () at (-2.4,5.4){$z_{1}$};

\node[v:main] () at (2,5){}; 
\node[v:ghost] () at (2.55,5.2){$w_{1}$};

\node[v:main] () at (-3,3){};
\node[v:ghost] () at (-3.5,3){$z_{2}$};

\node[v:main] () at (3,3){};
\node[v:ghost] () at (3.6,3){$w_{2}$};

\node[v:main] () at (-2,1){};
\node[v:ghost] () at (-2.35,0.65){$z_{p}$};

\node[v:main] () at (2,1){}; 
\node[v:ghost] () at (2.45,0.65){$w_{q}$};

\node[v:ghost] () at (-1.3,3){$\Gamma_{e,f_{1}}$};
\node[v:ghost] () at (1.3,3){$\Gamma_{e,f_{2}}$};

\begin{pgfonlayer}{background}

\draw[e:main,bend left,fill=blue!40] (0,6) to (-2,5) to (0,6){};

\draw[e:main,bend left,fill=blue!40] (0,6) to (2,5) to (0,6){};

\draw[e:main, bend left,fill=blue!40] (-2,5) to (-3,3) to (-2,5){};

\draw[e:main, bend left,fill=blue!40] (2,5) to (3,3) to (2,5){};

\draw[e:main,dashed,bend left] (3,3) to (2,1){}; 

\draw[e:main,dashed,bend right] (-3,3) to (-2,1){}; 

\draw[e:main,bend left,fill=blue!40] (0,0) to (-2,1) to (0,0){};

\draw[e:main,bend left,fill=blue!40] (0,0) to (2,1) to (0,0){};

\draw[e:main] (0,0) to (0,6){};

\end{pgfonlayer}
\end{tikzpicture}
\end{subfigure}
\begin{subfigure}{0.3\textwidth}
\centering
\begin{tikzpicture}[scale=0.5]

\pgfdeclarelayer{background}
\pgfdeclarelayer{foreground}
\pgfsetlayers{background,main,foreground}

\node[v:main] () at (0,6){};
\node[v:ghost] () at (0,6.4){$v$};

\node[v:main] () at (0,0){};
\node[v:ghost] () at (0,-0.4){$u$};


\node[v:main] () at (2,5){}; 
\node[v:ghost] () at (2.55,5.2){$w_{1}$};

\node[v:main] () at (-2,5){};
\node[v:main] () at (-2,1){};

\node[v:main] () at (3,3){};
\node[v:ghost] () at (3.6,3){$w_{2}$};


\node[v:main] () at (2,1){}; 
\node[v:ghost] () at (2.45,0.65){$w_{q}$};

\node[v:ghost] () at (1.3,3){$\Gamma_{e,f_{2}}$};

\begin{pgfonlayer}{background}

\draw[e:main] (-2,5) to (0,6){};

\draw[e:main,bend left,fill=blue!40] (0,6) to (2,5) to (0,6){};


\draw[e:main,bend left,fill=blue!40] (-2,1) to (-2,5) to (-2,1){};

\draw[e:main, bend left,fill=blue!40] (2,5) to (3,3) to (2,5){};

\draw[e:main,dashed,bend left] (3,3) to (2,1){}; 

%

\draw[e:main] (-2,1) to (0,0){};
\draw[e:main,bend left,fill=blue!40] (0,0) to (2,1) to (0,0){};

\draw[e:main] (0,0) to (0,6){};

\end{pgfonlayer}
\end{tikzpicture}
\end{subfigure}
\begin{subfigure}{0.3\textwidth}
\centering
\begin{tikzpicture}[scale=0.5]

\pgfdeclarelayer{background}
\pgfdeclarelayer{foreground}
\pgfsetlayers{background,main,foreground}

\node[v:main] () at (0,6){};
\node[v:ghost] () at (0,6.4){$v$};

\node[v:main] () at (0,0){};
\node[v:ghost] () at (0,-0.4){$u$};


\node[v:main] () at (2,5){}; 
\node[v:ghost] () at (2.55,5.2){$w_{1}$};


\node[v:main]() at (-3,3){};

\node[v:main] () at (3,3){};
\node[v:ghost] () at (3.6,3){$w_{2}$};


\node[v:main] () at (2,1){}; 
\node[v:ghost] () at (2.45,0.65){$w_{q}$};

\node[v:ghost] () at (1.3,3){$\Gamma_{e,f_{2}}$};

\begin{pgfonlayer}{background}


\draw[e:main,bend left,fill=blue!40] (0,6) to (2,5) to (0,6){};


\draw[e:main, bend left,fill=blue!40] (2,5) to (3,3) to (2,5){};

\draw[e:main,dashed,bend left] (3,3) to (2,1){}; 

%
\draw[e:main, bend left,fill=blue!40] (0,0) to (-3,3) to (0,0){};

\draw[e:main, bend left,fill=blue!40] (0,6) to (-3,3) to (0,6){};

\draw[e:main,bend left,fill=blue!40] (0,0) to (2,1) to (0,0){};

\draw[e:main] (0,0) to (0,6){};

\end{pgfonlayer}
\end{tikzpicture}
\end{subfigure}
\caption{In the left figure, the edge $e'=\{u,v\}$, the subgraphs $\Gamma_{e,f_{1}}$ and $\Gamma_{e,f_{2}}$ and the cut-vertices are depicted. In the central figure we see the form of the obstruction in Claims 5 and 6. In the right figure we see the form of the obstruction in Claim 7.}
\label{label_bouree}
\end{figure}
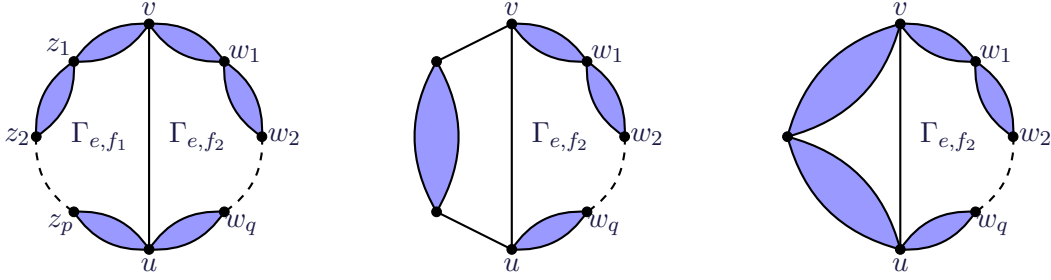

We consider an internal bidirectional edge, say $e=\{v_{f_{1}},v_{f_{2}}\}$. Let $e'=\{u,v\}$ be the edge in $E(B_{\Gamma}(f_{1}))\cap E(B_{\Gamma}(f_{2}))$. 
Observe that $e'$ belongs to the outer face of the graph $\Gamma_{e,f_{i}}$ and hence, $\Gamma_{e,f_{i}}\setminus e'$ {is not}, $i\in [2]$.
Recall that $\Gamma_{e,f_{1}}$ is outerplanar and let $z_{1},z_{2},\dots, z_{p}$, $p\geq 1$, denote the cut-vertices of $\Gamma_{e,f_{1}}\setminus e'$ according to the order they appear on the outer face of $\Gamma_{e,f_{1}}\setminus e'$ when traversing it from $v$ to $u$. Let also $\Gamma^{1}_{i}$, $i\in [p+1]$, denote the blocks that contain the vertices $z_{i-1}$ and $z_{i}$, where $z_{0}=v$ and $z_{p+1}=u$.
Similarly, let also $w_{1},w_{2},w_{q}$, $q\geq 1$ denote the cut-vertices of $\Gamma_{e,f_{2}}\setminus e'$ according to the order they appear on the outer face of $\Gamma_{e,f_{2}}\setminus \{u,v\}$ when traversing it from $v$ to $u$. Let also $\Gamma^{2}_{i}$, $i\in [q+1]$, 
denote the blocks that contain the vertices $w_{i-1}$ and $w_{i}$, where $w_{0}=v$ and $w_{p+1}=u$. The blocks $\Gamma^{1}_{i},\Gamma^{2}_{j}$, $i\in [p+1]$, $j\in [q+1]$ that do not contain $v$ or $u$ are called \emph{free}. (See \autoref{label_bouree})\medskip

\noindent{\em Claim 4}: It holds that $\mathsf{bed}_{\mathcal{G}}(H)=k+1$ for some $H\in \{\Gamma^{1}_{i},\Gamma^{2}_{j}\mid i\in [p+1], j\in [q+1]\}$.\medskip

\noindent{\em Proof of Claim 4:} Towards a contradiction assume that $\mathsf{bed}_{\mathcal{G}}(H) \leq k$ for every $H\in \{\Gamma^{1}_{i},\Gamma^{2}_{j}\mid i\in [p+1], j\in [q+1]\}$. Notice that every block $B$ of $\Gamma\setminus u$ is a subgraph of some block 
$\Gamma^{1}_{i}$, $i\in [p+1]$, or $\Gamma^{2}_{j}$, $j\in [q+1]$, and therefore, $\mathsf{bed}_{\mathcal{G}}(B)\leq k$. 
Then, from the definition of $\mathsf{bed}_{\mathcal{G}}$, $\mathsf{bed}_{\mathcal{G}}(\Gamma)\leq k+1$, a contradiction to the hypothesis
that $\Gamma \in \mathsf{obs}(\mathcal{G}^{(k+1)})$.\hfill$\diamond$\medskip

\noindent{\em Claim 5:} If there exists a free block $H\in \{\Gamma^{1}_{i},\Gamma^{2}_{j}\mid i\in [p+1], j\in [q+1]\}$ such that $\mathsf{bed}_{\mathcal{G}}(H)=k+1$ then the Lemma holds.\medskip

\noindent{\em Proof of Claim 5:} Indeed without loss of generality let $H=\Gamma^{1}_{i_{0}}$ for some $i_{0}\in [p+1]$.
Let us consider the graph obtained by contracting all vertices of $\Gamma_{e,f_{1}}$ except for $V(\Gamma^{1}_{i_{0}})$ and $\{u,v\}$. Observe then that the resulting graph $\Gamma'$ can be expressed as the parallel join of $\Gamma^{1}_{i_{0}}$ and $\mathsf{bed}_{\mathcal{G}}(\Gamma_{e,f_{2}})$
in the following way: $\Gamma'=||(\Gamma^{1}_{i_{0}},z_{i_{0}-1},z_{i_{0}},\Gamma_{e,f_{2}},v,u)$. From Claim 1,
since $e$ is bidirectional, we obtain that $\mathsf{bed}_{\mathcal{G}}(\Gamma_{e,f_{2}})=k+1$. Moreover, since $\mathsf{bed}_{\mathcal{G}}(\Gamma^{1}_{i_{0}})=k+1$, from \autoref{label_conflictivitat}, we obtain
that $\mathsf{bed}_{\mathcal{G}}(\Gamma')=k+2$. As $\Gamma \in \mathbf{obs}(\mathcal{G}^{(k+1)})$ and $\Gamma'$ is a minor of $\Gamma$, it follows that $\Gamma'=\Gamma$. From \autoref{label_asturias}, $\Gamma^{1}_{i_{0}}$, $\Gamma_{e,f_{2}}\in \mathbf{obs}(\mathcal{G}^{(k)})$ and
this indeed proves the statement of the Lemma.\hfill$\diamond$ \medskip

Therefore, from now on we will assume that if $H\in \{\Gamma^{1}_{i},\Gamma^{2}_{j}\mid i\in [p+1], j\in [q+1]\}$ and $\mathsf{bed}_{\mathcal{G}}(H)=k+1$ then $H$ is not free, that is, $H$ contains $u$ or $v$.\medskip

\noindent{\em Claim 6:} If all blocks $H\in \{\Gamma^{1}_{i},\Gamma^{2}_{j}\mid i\in [p+1], j\in [q+1]\}$ for which $\mathsf{bed}_{\mathcal{G}}(H)=k+1$ contain $v$, then the Lemma holds. Symmetrically, if all blocks $H\in \{\Gamma^{1}_{i},\Gamma^{2}_{j}\mid i\in [p+1], j\in [q+1]\}$ for which $\mathsf{bed}_{\mathcal{G}}(H)=k+1$ contain $u$, then the Lemma holds.\medskip

\noindent{\em Proof of Claim 6:} Observe that the only blocks that contain $v$ are $\Gamma^{1}_{1}$ and $\Gamma^{2}_{1}$. Moreover, for any other block $B$ of $\Gamma\setminus u$, it holds that 
$\mathsf{bed}_{\mathcal{G}}(B)\leq k$. Since $\Gamma \in\mathsf{obs}(\mathcal{G}^{(k+1)})$, it follows that $\mathsf{bed}_{\mathcal{G}}(\Gamma \setminus u)=k+1$. From the above discussion
and the definition of $\mathsf{bed}_{\mathcal{G}}$, we obtain that there exists a block $D$ of $\Gamma \setminus u$ for which $\mathsf{bed}_{\mathcal{G}}(D)=k+1$ and 
$D\subseteq \Gamma^{1}_{1}\setminus v$ or $D\subseteq\Gamma^{2}_{1}\setminus v$. Without loss of generality let us assume that $D\subseteq \Gamma^{1}_{1}\setminus v$. (The case where 
$D\subseteq\Gamma^{2}_{1}\setminus v$ is symmetrical). Let  $v_{1}$ denote the neighbor of  $v$ in $\Gamma_{e_{f_{1}}}$ that is also a neighbor of $v$ in the graph $B_{\Gamma}(f_{o})$. 
Observe also that $u\notin V(D)$. Let $\Gamma'$ be the graph obtained after we contract all blocks of $\Gamma_{e,{f_{1}}}$ except for the block $\Gamma^{1}_{1}$ to the edge $z_{1}u$
and remove all edges that contain $v$ in $E(\Gamma^{1}_{1})$ apart from the edge $v_{1}v$. Then $\Gamma'$ can be expressed as the parallel join of the graphs $\Gamma^{1}_{1}\setminus v$
and $\Gamma_{e,f_{2}}$ in the following way: $\Gamma'=||(\Gamma^{1}_{1},v_{1},z_{1},\Gamma_{e,f_{2}},v,u)$. From \autoref{label_conflictivitat}, we obtain that 
$\mathsf{bed}_{\mathcal{G}}(\Gamma')=k+2$. As $\Gamma'$ is a minor of $\Gamma$ and $\Gamma\in \mathsf{obs}(\mathcal{G}^{(k+1)})$, it follows that $\Gamma'=\Gamma$. Moreover,
 from \autoref{label_asturias}, $\Gamma^{1}_{1}\setminus v$ and $\Gamma_{e,f_{2}}$ belong to $\mathsf{obs}(\mathcal{G}^{(k)})$. Then indeed the Lemma holds and this concludes the proof of the claim.\hfill$\diamond$\medskip

\noindent{\em Claim 7:} If there exist two blocks $H,H'\in \{\Gamma^{1}_{i},\Gamma^{2}_{j}\mid i\in [p+1], j\in [q+1]\}$ such that $H$ contains $v$, $H'$ contains $u$, and $\mathsf{bed_{\mathcal{G}}}(H)=\mathsf{bed}_{\mathcal{G}}(H')=k+1$  then the Lemma holds. \medskip

\noindent{\em Proof of Claim 7:} We first examine the case where $H=\Gamma^{1}_{1}$ and $H'=\Gamma^{1}_{p+1}$. (The case where $H=\Gamma^{2}_{1}$ and $H'=\Gamma^{2}_{q+1}$ is symmetrical.)
Let $\Gamma'$ be the graph obtained from $\Gamma$ after contracting the vertices of all blocks $\Gamma^{1}_{2},\dots, \Gamma^{1}_{p}$ into a new vertex $y$. Observe that $\Gamma'$ can be expressed as triangular gluing of $\Gamma^{1}_{1}$, $\Gamma^{1}_{p+1}$, and $\Gamma_{e,f_{2}}$
in the following way: $\Gamma'=\triangle(\Gamma^{1}_{1},v,y,\Gamma^{1}_{p+1},y,u,\Gamma_{e,f_{2}},u,v)$. From Claim 1, we obtain that $\mathsf{bed}_{\mathcal{G}}(\Gamma_{e,f_{2}})=k+1$. Hence, from \autoref{label_forsakenness}, it follows that $\mathsf{bed}_{\mathcal{G}}(\Gamma')≥k+2$.
As $\Gamma \in \mathbf{obs}(\mathcal{G}^{(k+1)})$ and $\Gamma'$ is a minor of $\Gamma$, it follows that $\Gamma'=\Gamma$. From \autoref{label_asturias}, we obtain that the graphs $\Gamma^{1}_{1}$, $\Gamma^{1}_{p+1}$, and $\Gamma_{e,f_{2}}$ belong to $\mathbf{obs}(\mathcal{G}^{(k)})$. In this
case we have proven the assertion of the Lemma.

We now examine the case where $H=\Gamma^{1}_{1}$ and $H'=\Gamma^{2}_{q+1}$. (The case where $H=\Gamma^{2}_{1}$ and $H'=\Gamma^{1}_{p+1}$ is symmetrical.) Let $\Gamma'$ be the graph obtained from $\Gamma$ after contracting the edges of all blocks 
$\Gamma^{1}_{2},\dots, \Gamma^{1}_{p+1}$ into the single edge $z_{1}u$ and the edges of all blocks $\Gamma^{2}_{1},\dots, \Gamma^{1}_{q}$ except $w_{q}v$ and finally removing the edge $e'=uv$. Observe that $\Gamma'$ can be expressed as the parallel join
of $\Gamma^{1}_{1}$ and $\Gamma^{2}_{q+1}$ in the following way: $\Gamma'=||(\Gamma^{1}_{1},v,z_{1},\Gamma^{2}_{q+1},w_{q},u)$. Observe that $\Gamma'$ is a proper minor of $\Gamma$ and from \autoref{label_conflictivitat}, $\mathsf{bed}_{\mathcal{G}}(\Gamma')≥k+2$. This is a contradiction
to the hypothesis that $\Gamma\in \mathsf{obs}(\mathcal{G}^{(k+1)})$. This concludes the proof of the claim.\hfill$\diamond$\medskip

The above claims complete the proof of the Lemma.
\end{proof}

\autoref{label_conflictivitat}, \autoref{label_forsakenness}, and \autoref{label_desencantado} the following.

\begin{theorem}
\label{dg_oplo}
For every non-trivial minor-closed class ${\cal G}$ and every $k\in \mathbb{N}$, 
every outerplanar graph in ${\sf obs}({\cal G}^{(k+1)})$ can be generated by applying 
either the parallel join or the triangular gluing operation to outerplanar graphs of ${\sf obs}({\cal G}^{(k)})$ in a way that preserves outerplanarity.
\end{theorem}

As ${\bf obs}({\cal G}^{(0)})={\sf obs}({\cal B}({\cal G}))$, \autoref{label_trivialities} and \autoref{dg_oplo} give  a complete characterization of ${\cal O}\cap{\cal G}^{(k)}$, for every 
$k\in\mathbb{N}$ and every non-trivial minor-closed graph class ${\cal G}$. It is easy to verify that for every ${\cal G}$,   there are at least two obstructions in ${\sf obs}{({\cal G}^{(3)})}$ that are generated by the triangular gluing operation.
Moreover, as the operation of trianglular gluing three graphs from a set of $q$  graphs results to $q^{2}+{q\choose 3}\geq q^{2}$ new graphs, our results imply that, for $k\geq 3$,  $|{\sf obs}{({\cal G}^{(k)})}|\geq |{\sf obs}{({\cal G}^{(k-1)})}|^2$.
It follows  that, for every non-trivial minor-closed class $\mathcal{G}$, ${\sf obs}{({\cal G}^{(k)})}$ contains doubly  exponentially many graphs. 

\section{A conjecture on the  universal obstructions}
\label{label_siberia}

Recently, Huynh et al. in~\cite{huynh2020excluding}  defined the parameter ${\sf td}_{2}$ as follows. A {\em biconnected centered coloring}
of a graph $G$ is a vertex coloring of $G$ such that for every connected subgraph $H$ of $G$
that is a block graph,  some color is assigned to {\em exactly one} vertex of $H$. 
Given a non-empty graph $G$,  ${\sf td}_{2}(G)$  is defined as the minimum number of colors
in a biconnected centered coloring of $G$. Using the alternative definition of \autoref{dr_oteros},
it can easily be verified that, for every non-empty graph $G$, ${\sf td}_{2}(G)={\sf bed}_{\cal E}(G)+1$. 
We define the {\em $t$-ladder} as the  $(2\times t)$-grid
(i.e., the Cartesian product of $K_{2}$ and a path on $t$-vertices) and we denote it by $L_{t}$.
It is easy to check that ${\sf td}_{2}(L_{t})={\rm Ω}(\log(t))$. One of the main results of~\cite{huynh2020excluding}
was that there is a function $f:\mathbb{N}\to\mathbb{N}$
such that  every graph excluding a $t$-ladder belongs to ${\cal E}^{(f(t))}$.
This implies that the $t$-ladder $L_{t}$ is a {\sl universal minor} obstruction    
for ${\sf bed}_{\cal E}$. This motivates us to make a conjecture on how the results of~\cite{huynh2020excluding}
should be extended for every non-trivial minor-closed class  ${\cal G}$: 
Given a positive $t$, we define ${\cal L}_{{\cal G},t}$ as the class containing every graph that can be 
constructed by first taking the disjoint union of two paths $P_i,i\in[2]$, with vertices $v_{1}^{i},\ldots,v_{t}^i$ (ordered the way they appear in $P_i$) and $t$ graphs $G_{1},\ldots,G_{t}$ from $\obs({\cal B}({\cal G}))$ and then, for $i\in[t]$,  identify $v_{i}^{1}$ and $v_{i}^{2}$
with two different vertices in $G_{i}$. It is easy to check that if $G\in {\cal L}_{{\cal G},t}$, then ${\sf bed}_{\cal G}(G)={\rm Ω}(\log t)$.
We conjecture that ${\cal L}_{{\cal G},t}$ is a  {\sl universal minor obstruction}
for ${\sf bed}_{\cal G}$, i.e., there is a  function $f:\mathbb{N}\to\mathbb{N}$
such that  every graph excluding all graphs in  ${\cal L}_{{\cal G},t}$ as a minor, has block elimination distance 
to ${\cal G}$ bounded by $f(t)$, i.e., 
${\sf excl}({\cal L}_{{\cal G},t})\subseteq {\cal G}^{(f(t))}$. 
Notice that the two operations of \autoref{dg_oplo} imply that, when restricted to outerplanar graphs, this conjecture is correct for $f(t)=O(t)$. However we do not believe that the linear upper bound is maintained in the general case.
\medskip

\noindent {\bf Acknowledgements:} Öznur Yaşar Diner is grateful to the 
members of the research group \mbox{\sf GAPCOMB} for hosting a research stay at 
Universitat Politècnica de Catalunya.

%
%

\end{document}